\documentclass[10pt,a4paper,reqno]{amsart} 
\usepackage{mathrsfs}

\usepackage{amsmath}
\usepackage{bbm}
\usepackage{mathtools}
\usepackage{amssymb}
\usepackage{graphicx}
\usepackage{color}
\usepackage{latexsym}
\usepackage[utf8]{inputenc}
\usepackage[T1]{fontenc}
\usepackage{comment}
\usepackage{stmaryrd}

\usepackage{caption}
\usepackage{subcaption}

\newcommand{\ns}{{\mathbb N}} 
\newcommand{\zs}{{\mathbb Z}} 
\newcommand{\cs}{{\mathbb C}} 

\newcommand{\eps}{\varepsilon}

\newcommand{\cA}{\mathcal A}
\newcommand{\cB}{\mathcal B}

\newcommand{\cE}{\mathcal E}
\newcommand{\bcE}{\bar{\mathcal E}}

\def\x{x_{\rm c}}
\def\y{y_{\rm c} }
\def\ys{y^*}

\def\e{{\rm e}}
\def\d{{\rm d}}
\def\ii{i}

\newcommand{\stickbreak}{\mathsf{StickBreak}}
\newcommand{\height}{\mathsf{H}}
\newcommand{\width}{\mathsf{W}}
\newcommand{\bbH}{\mathbb H}

\newcommand{\prneinf}{\PR_{{\rm iSAB}}^{\otimes\mathbb N}}
\newcommand{\prneinfbi}{\PR_{{\rm iSAB}}^{\otimes \mathbb Z}}

\newcommand{\ner}{\mathsf{R}}

\newcommand{\sabset}{{\rm SAB}}

\newcommand{\isab}{{\rm iSAB}}
\newcommand{\prne}{\PR_{{\rm iSAB}}}

\newcommand{\PR}{\mathbb{P}}

\newcommand{\xco}{\mathbf x}
\newcommand{\yco}{\mathbf y}

\newtheorem{Theorem}{Theorem}
\newtheorem{Proposition}[Theorem]{Proposition}

\newtheorem{Corollary}[Theorem]{Corollary}
\newtheorem{Conjecture}[Theorem]{Conjecture}

\newtheorem{Lemma}[Theorem]{Lemma}

\newcommand{\beq}{\begin{equation}}
\newcommand{\eeq}{\end{equation}}

\newcommand{\gf}{generating function}
\newcommand{\gfs}{generating functions}

\newcommand{\saw}{self-avoiding walk}
\newcommand{\saws}{self-avoiding walks}
\def\emm#1,{{\em #1}}

\newcommand{\red}[1]{\textcolor{red}{#1}}

\catcode`\@=11
\def\section{\@startsection{section}{1}%
 \z@{.7\linespacing\@plus\linespacing}{.5\linespacing}%
 {\normalfont\bfseries\scshape\centering}}

\def\subsection{\@startsection{subsection}{2}%
  \z@{.5\linespacing\@plus\linespacing}{.5\linespacing}%
  {\normalfont\bfseries\scshape}}

\def\subsubsection{\@startsection{subsubsection}{3}%
 \z@{.5\linespacing\@plus\linespacing}{-.5em}
  {\normalfont\bfseries\itshape}}
\catcode`\@=12

\def\qed{$\hfill{\vrule height 3pt width 5pt depth 2pt}$}

\newcommand{\spacebreak}
{\begin{displaymath} \triangleleft \; \lhd \;
\diamond \; \rhd \; \triangleright
  \end{displaymath}}

\begin{document}
\title
[The critical fugacity for surface adsorption of  self-avoiding walks]
{The critical fugacity for surface adsorption of self-avoiding walks on the honeycomb lattice is $\mathbf{1+\sqrt{2}}$}
\author[N. Beaton]{Nicholas R. Beaton}
\author[M. Bousquet-M\'elou]{Mireille Bousquet-M\'elou}
\author[J. de Gier]{Jan de Gier}
\author[H. Duminil-Copin]{Hugo Duminil-Copin}
\author[A. J. Guttmann]{Anthony J. Guttmann}

\address{NRB: Laboratoire d'Informatique de Paris Nord, Universit\'e Paris 13, 93430 Villetaneuse, France}
\email{nicholas.beaton@lipn.univ-paris13.fr}
\address{MBM: CNRS, LaBRI, UMR 5800, Universit\'e de Bordeaux, 
351 cours de la Lib\'eration, 33405 Talence Cedex, France}
\email{mireille.bousquet@labri.fr}
\address{HDC: Section de Math\'ematiques, Universit\'e de Gen\`eve,
Gen\`eve, Switzerland}
\email{hugo.duminil@unige.ch}
\address{JdG, AJG: Department of Mathematics and Statistics, The
  University of Melbourne, VIC 3010, Melbourne, Australia}
\email{jdegier,guttmann@ms.unimelb.edu.au}
%

\thanks{}

\subjclass[2000]{}

\begin{abstract}
In 2010, Duminil-Copin and Smirnov proved a long-standing conjecture
of Nienhuis, made in 1982, that the growth constant 
of self-avoiding walks on the hexagonal 
({a.k.a.} honeycomb) lattice is $\mu=\sqrt{2+\sqrt{2}}.$ A
key identity used in that proof was later generalised by Smirnov so as
to apply to a general $O(n)$ loop model with $n\in [-2,2]$
(the case $n=0$ corresponding to \saws).

We modify this model by restricting to a half-plane and introducing a
surface fugacity $y$ associated
with boundary sites (also called surface sites), and obtain a generalisation of Smirnov's identity.
The critical value of the surface fugacity was conjectured by
Batchelor and Yung in 1995 to be $y_{\rm c}=1+2/\sqrt{2-n}.$ This
value plays a crucial role in our generalized identity, just as the
value of growth constant did in {Smirnov's identity.}

For the case $n=0$, corresponding to \saws\ interacting with a surface,
we prove the conjectured value of the critical surface
fugacity. A {crucial}
 part of the proof involves demonstrating that
the generating function of self-avoiding bridges of height $T$, taken
at its critical point $1/\mu$, tends to $0$ as $T$ increases, as
predicted from SLE theory. 
\end{abstract}

\date{\today}
\maketitle

\section{Introduction}
\label{sec:intro}

The $n$-vector model, also called $O(n)$ model, introduced by Stanley
in 1968 \cite{S68} is described by the Hamiltonian 
$$
{\mathcal H}(d,n) = -J\sum_{\langle i, j \rangle} {\bf s}_i \cdot{\bf  s}_j ,
$$
where $d$ denotes the dimensionality of the lattice, $i$ and $j$ are
adjacent sites, and ${\bf s}_i$ is an $n$-dimensional vector of
magnitude $\sqrt n$. Writing $x=J/k_B T$, the corresponding
  partition function of this model on a 
{two-dimensional  square domain} with $N^2$ sites is  given by 
\begin{equation}
Z_{N^2} {(x)}
=\int \prod_k \d \mu({\bf s}_k) \prod_{\langle i, j \rangle} w_{ij},\qquad w_{ij}= \e^{x {\bf s}_i \cdot{\bf s}_j}, 
\label{eq:partsum}
\end{equation}
where   $\mu$ is the spherical measure on the $(n-1)$-dimensional
sphere of radius $\sqrt n$, normalised by $\int \d\mu({\bf s})=1$.

When $n=1$ {the} Hamiltonian {above}
describes the Ising model, and when $n=2$ it describes the classical
XY model. Two other interesting limits, which leave a lot to be
desired from a pure mathematical perspective, are the limit $n \to 0,$
in which case one recovers the self-avoiding walk (SAW) model, as
first pointed out by de Gennes~\cite{dG72}; and the limit $n \to -2$,
corresponding to random walks, or more generally to
{a free-field Gaussian}
model, as shown by Balian and Toulouse~\cite{BT73}.

Self-avoiding walks will be central in 
this paper. They 
have been considered as models of long-chain polymers in solution
since the middle of the last century --- see for 
example articles by Orr~\cite{O47} and Flory~\cite{F53}. Since that
time they have been studied and {extended} 
by polymer chemists as models
of polymers; by mathematicians as combinatorial models of pristine
simplicity in their description, yet malevolent difficulty in their
solution;
by computer scientists interested in computational
complexity; and by biologists using them to model properties of DNA
and other biological polymers of interest.  

It is argued in~\cite{DMNS81} that the critical
  behaviour of the  $n$-vector model is unchanged if the   Boltzmann weight $w_{ij}$ in
  \eqref{eq:partsum} is replaced by  $w_{ij}=1+x\, {\bf s}_i \cdot{\bf
    s}_j$. Moreover, this new model  is
  equivalent to a loop model with a weight $n$ attached to 
closed loops~\cite{DMNS81}:
$$
Z_{N^2}(x)=\sum_{\gamma}x^{|\gamma|}n^{\ell(\gamma)},
$$
where {(on the honeycomb lattice)} 
$\gamma$ is a configuration of non-intersecting loops,
$|\gamma|$ is the number 
of edges and $\ell(\gamma)$ is the number of  loops.
{We call this model the $O(n)$ \emm loop model,.}
In the following we consider 
a loop model with a defect,
  {\it i.e.}, a model of closed loops with one \saw\
  component\footnote{Defects correspond to correlation functions of
    the underlying spin model. It follows that the critical point
    remains the same.}. A typical configuration  is
shown in Fig.~\ref{fig:configuration}, left. 

\begin{figure}
\centering
\includegraphics[scale=0.4]{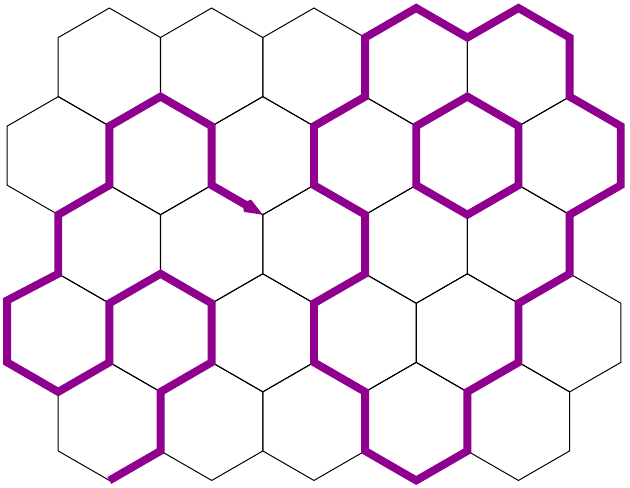}
\hspace{1cm}
\includegraphics[scale=0.87]{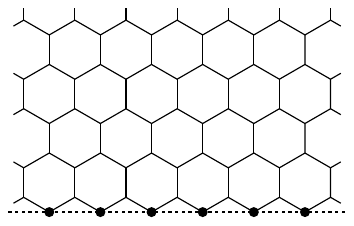}
\caption{\emm Left,: A configuration of the loop model on the honeycomb
  lattice. 
{\emm Right,: The {half-plane and its boundary.}}}
\label{fig:configuration}
\end{figure}

In 1982 Nienhuis \cite{N82} showed that, for $n \in [-2,2],$ the
loop model on the honeycomb lattice could be mapped onto a solid-on-solid model,
from which he was able to derive the critical points and
critical exponents, subject to some plausible assumptions. These
results agreed with the known exponents and critical point for the
Ising model, and they predicted exact values for those models corresponding
to other values of the spin dimensionality $n$. 
In particular, for $n=0$ the 
{critical point} for the honeycomb lattice
SAW model was predicted to be $\x =1/\sqrt{2+\sqrt{2}}$.

This result was finally proved 28 years later by  Duminil-Copin and
Smirnov~\cite{DC-S10}. 
The starting point of their proof is  a
\emm local,\, identity for a ``parafermionic'' observable, 
valid at every vertex of the lattice.  
Then they obtain a \emm global,\,
identity linking several walk \gfs\  by summing over all vertices
of a domain\footnote{A more formal presentation of their proof has 
 been provided by Klazar~\cite{Klaz}.}. 
Smirnov~\cite{Smirnov10} then extended
the local identity to the general honeycomb $O(n)$ loop
 model with $n \in [-2,2]$. This extension 
 provides an alternative way of predicting the
value of the  critical point $\x (n)=1/\sqrt{2+\sqrt{2-n}}$ as
conjectured by Nienhuis.

\medskip
Nienhuis's results were concerned with bulk systems. Interesting surface
phenomena can also be studied if one considers the {$O(n)$ loop} model in
a half-space, with vertices in the surface 
(the boundary of the half-space)
having an associated fugacity. 
{See Fig.~\ref{fig:configuration}, right.}
The partition function becomes
\begin{equation}\label{eqn:loop_model_partition}
Z_{N^2}(x,y) = \sum_\gamma x^{|\gamma|} y^{c(\gamma)} n^{\ell(\gamma)},
\end{equation}
where $c(\gamma)$ is the number of vertices on the boundary occupied
by $\gamma$. The value of the fugacity $y$ can be changed to result in
a repulsive or attractive interaction with the surface, with a phase
transition occurring at the point 
distinguishing {between}
 these two regimes. In
the limit of a large lattice, the free energy per site may be
decomposed as 
$$
f_{N^2}(x,y) := -\frac{k_B T}{N^2} \log Z_{N^2} (x,y)= f_{\rm bulk}(x,y) + \frac{1}{N} f_{\rm surface}(x,y) +\cdots,
$$
and surface phase transitions correspond to singularities in $f_{\rm
  surface}$.
{At $x=\x$, the adsorption transition is an example of a \emm special,
surface transition \cite{Binder83}.   }

In 1995 Batchelor and Yung~\cite{BY95} extended Nienhuis's work to the
adsorption problem described above. {Using the integrability of an
  underlying lattice model and comparison to numerical results, they}
conjectured the value of
the critical surface fugacity for the  honeycomb lattice {$O(n)$
  loop} model.\footnote{{Batchelor and Yung use different
    notation, and in particular weight their configurations slightly
    differently. They use weights $t_b$ and $t_s$ (corresponding to
    bulk and surface edges respectively), and the correspondence with
    our notation is $x=1/t_b$ and $y=t_b/t_s$.}} 

\begin{Conjecture}[\bf{Batchelor and Yung}] 
\label{conj1}
For the $O(n)$ loop model on the semi-infinite  hexagonal lattice 
{of Fig.~\rm{\ref{fig:configuration}}}
 with $n\in[-2,2]$, associate a fugacity $\x(n) =
 1/\sqrt{2+\sqrt{2-n}}$ with occupied vertices and an additional
 fugacity $y$ with occupied vertices on the boundary. Then the model
 undergoes a surface transition at 
\[y = \y(n) = 1+\frac{2}{\sqrt{2-n}}.\]
\end{Conjecture}

 In this paper we first show that the local
 identity proved by Smirnov~\cite{Smirnov10} for the
$O(n)$ loop model can be generalised to a half-plane
 system with a surface fugacity (Lemma~\ref{lem:local}). 
We use this to prove a  generalisation  of the global
identity of   Duminil-Copin and Smirnov 
 including a surface fugacity 
(Proposition~\ref{prop:global}). 
The contribution of one of these generating functions vanishes at 
$y=\y (n)$, which lends support to the above conjecture.

We then focus on the case $n=0$, corresponding to \saws\ interacting
with an impenetrable surface. 
{This case is somewhat degenerate since no loops are allowed and
  we therefore adopt the definition of $y_c(0)$ given by Hammersley,
  Torrie and Whittington~\cite{HTW82}, which we recall in Section 3.1.}
  With this definition, we prove Conjecture~\ref{conj1} for $n=0$:
 a \saw\ is adsorbed if $y>1+\sqrt 2$ and desorbed if $y<1+\sqrt 2$. 

\begin{Theorem}\label{thm:THEOREM} 
{The critical surface fugacity for self-avoiding walks on the honeycomb lattice is}
\[{\y(0)} = 1+\sqrt{2}.\]
\end{Theorem}

The proof of Theorem~\ref{thm:THEOREM} relies of course on our global
identity, but also requires earlier results dealing with SAWs 
confined to a half-plane or a strip: {notably,} existence of {the
  critical value} of the surface fugacity $y$, 
enumeration of SAWs in a strip and  the
behaviour  as the size of the
strip increases, among others. Most of these results have been proved for the
square (and hypercubic) lattice, but we need to adapt these proofs to
the honeycomb case, which we do in Section~\ref{sec:strip}. Section~\ref{sec:proof}
combines these results and the global identity to prove Theorem~\ref{thm:THEOREM}. A third key
ingredient, of independent interest, is that the \gf\ of \emm bridges,
of height $T$, taken at $\x $, tends to 0 as $T$ increases. The proof
is probabilistic in nature, and is given in the appendix.

\medskip
To conclude this introduction, 
let us mention that one can also consider a honeycomb
half-plane with a \emm vertical, (rather than horizontal)
boundary. The techniques of this paper have been adapted by the first
author to determine the critical surface fugacity in this case~\cite{beaton-rotated},
which was conjectured by  Batchelor, Bennett-Wood and Owczarek~\cite{BBO}.
For other lattices, we do not 
 have conjectures for the values of the critical
fugacities; instead, numerical estimates using
series analysis and Monte Carlo methods are the best current
results. New methods of estimating the growth constants and
critical surface fugacities of the square and triangular lattices,
inspired by results presented in~\cite{DC-S10} and this paper, are
explored in~\cite{BGJ2011a,BGJ2011b}.

\section{Smirnov's identity in the presence of a boundary}
\label{sec:identity}

\begin{figure}
\centering
\includegraphics[scale=0.4]{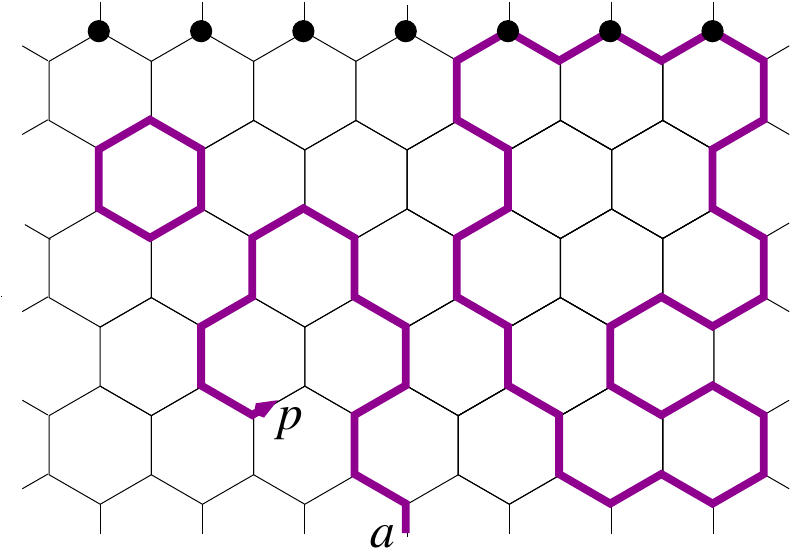}
\caption{A configuration $\gamma$ on a finite domain, with the
  weighted vertices on the top  boundary indicated. The
  contribution of $\gamma$ to {$F(p)$}
is  $\e^{-5\ii\sigma\pi/3}x^{51}y^3n^2$.  }
\label{fig:exampleF}
\end{figure}

We consider the honeycomb lattice, embedded in the complex plane $\cs$
in such a way that the edges have unit length. This  allows us  to
consider vertices of the lattice as complex numbers. 
It is also convenient to start and end \saws\ at a
mid-edge of the lattice. {(That is, the point on an edge precisely halfway between its two incident vertices.)} We restrict the lattice to a half-plane,
bounded 
{from above}
by a horizontal surface consisting of \emm weighted sites,
(Fig.~\ref{fig:exampleF}).
We further consider a  domain $D$ of this half-lattice, consisting
of a finite connected collection of half-edges 
such that for every vertex $v$ incident to at least one half-edge of
$D$, all three half-edges incident to $v$ actually belong to
$D$.
We denote by $V(D)$ the set of vertices incident to half-edges of
$D$. Those mid-edges of $D$ which are adjacent to only one
vertex in $V(D)$ form the \emm boundary, $\partial D$. 
A \emm configuration, $\gamma$ consists of a (single) \saw\ $w$ and a (finite)
collection of closed loops, which are self-avoiding and do not meet one
another nor $w$.
We denote by
$|\gamma|$ the number of vertices occupied by $\gamma$ (also called
the \emm length,), by $c(\gamma)$ the
number of \emm contacts,  with the surface (\emph{i.e.} vertices of the
surface occupied by $\gamma$), and by $\ell(\gamma)$ the number of 
 loops. See Fig.~\ref{fig:exampleF} for an example. 

{Let $a$ be a fixed mid-edge on the boundary $\partial D$.}
For any mid-edge $p$ of $D$, 
define the following \emm \gf,, or \emm observable,\/: 
\beq\label{F-def}
F(
p;x,y,n,\sigma)\equiv F(p) := 
\sum_{\gamma : a\leadsto p} 
 x^{|\gamma|} y^{c(\gamma)}n^{\ell(\gamma)}\e^{-\ii \sigma W(w)},
\eeq
where the sum is over all configurations $\gamma$ in $D$ for
which the SAW component $w$ runs from $a$ to $p$ and 
$W(w)$ is the \emm winding angle, of $w$,
that is, $\pi/3$ times the difference between the number of left turns
and the number of right turns. 

The case $y=1$ of the following lemma is due to Smirnov~\cite{Smirnov10}.
\begin{Lemma}[{\bf {The local identity}}]
\label{lem:local}
For $n\in[-2,2]$, set $n=2\cos\theta$ with $\theta\in[0,\pi]$. Let
\begin{align}
\sigma &= \frac{\pi-3\theta}{4\pi},\qquad 
{x^{-1}} = 
\x^{-1}:=2\cos\left(\frac{\pi+\theta}{4}\right) = \sqrt{2-\sqrt{2-n}},\qquad\text{or}
\label{eq:Slemma_dense}\\
\sigma &= \frac{\pi+3\theta}{4\pi},\qquad 
x^{-1}=\x^{-1}:=2\cos\left(\frac{\pi-\theta}{4}\right) = \sqrt{2+\sqrt{2-n}}.
\label{eq:Slemma_dilute}
\end{align}
Then for a vertex $v\in V(D)$
not belonging to the weighted surface, the observable $F$ defined
by~\eqref{F-def} satisfies
\begin{equation} \label{eqn:localidentity}
(p-v)F(p) + (q-v)F(q) + (r-v)F(r)=0,
\end{equation}
where $p,q,r$ are the mid-edges adjacent to $v$.

If $v\in V(D)$ lies on the weighted surface,
\begin{multline}
  (p-v)F(p) + (q-v)F(q) + (r-v)F(r)= \\
(q-v)(1-y)(\x y\lambda)^{-1}\sum_{\gamma : a \leadsto q, p} \x ^{|\gamma|} y^{c(
  \gamma)} n^{\ell(\gamma)} \e^{-\ii \sigma W(w)}\label{surface-case}\\
+(r-v)(1-y) (\x y\bar \lambda)^{-1}\sum_{\gamma :  a\leadsto r, p} \x ^{|\gamma|} y^{c(
  \gamma)} n^{\ell(\gamma)} \e^{-\ii \sigma W(w)},
\end{multline}
where $\lambda=\e^{-\ii\sigma\pi/3}$ is the weight accrued by a walk for
each left turn, $p,q,r$ are the three mid-edges adjacent to $v$, taken
in counterclockwise order, with $p$ just above $v$, and the first
(resp. second) sum runs over configurations $\gamma$ {whose SAW
  component $w$ goes} from $a$ to $p$ via $q$ (resp. via $r$).
\end{Lemma}
Equation~\eqref{eq:Slemma_dense} corresponds to the larger of the two special
values of the step weight $x$ and corresponds to a point in the dense
regime where the model turns out to be
integrable. Equation~\eqref{eq:Slemma_dilute} gives the value of the
critical point, separating the dense and dilute phases. Both special
values of $x$ were predicted by Nienhuis~\cite{N82} from a
renormalisation group analysis, the first value corresponding to a
stable fixed point and the latter to an unstable one. 
In what follows, 
when we refer to the dense and dilute regimes, we mean the regimes with values of the step weight $x$ given by~\eqref{eq:Slemma_dense} and~\eqref{eq:Slemma_dilute} respectively.

\begin{figure}
\centering
\scalebox{0.7}{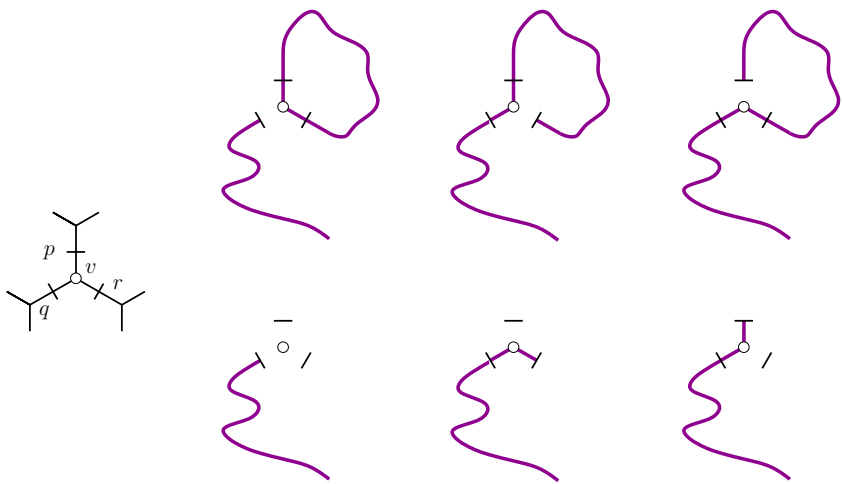}
\caption{Two groups of  configurations  ending at a mid-edge
adjacent to the vertex $v$. 
{The contribution of each group to~\eqref{eqn:localidentity} is 0.}}
\label{fig:identity_groups} 
\end{figure}

\begin{proof}
If $v$ does not belong to the surface, the proof is completely
analogous to the proof of Lemma~4 in~\cite{Smirnov10}: One observes that the left-hand side
of~\eqref{eqn:localidentity} counts (weighted) configurations ending
at a  mid-edge adjacent to $v$, and organizes these configurations by
groups of three, as shown in  Fig.~\ref{fig:identity_groups} (which,
up to rotations, includes all possible cases). It is
then easy to check that, for the given values of $\sigma$ and $\x $,
the contribution of each group vanishes. The fact that
$y\not = 1$ in our paper makes no difference, because  the
number of weighted vertices is the same for all walks  in a
group.

\begin{figure}
\centering
\scalebox{0.7}{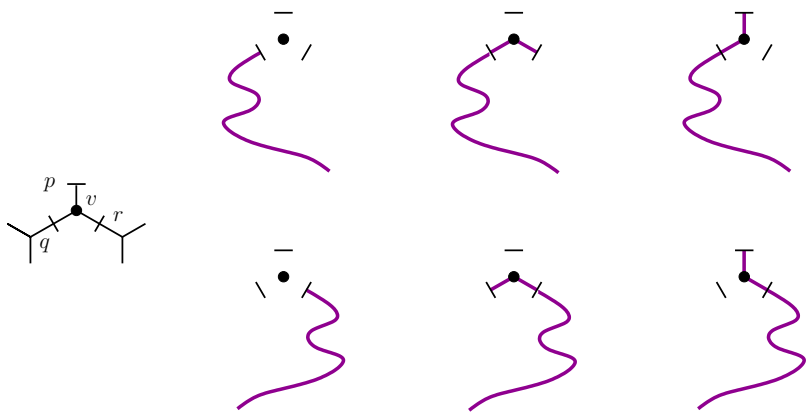}
\caption{
{Two groups of walks ending at a mid-edge adjacent to a 
surface  vertex. The top 
group leads to~\eqref{eqn:boundaryeqn_first}, and  the bottom 
group   to~\eqref{eqn:boundaryeqn_second}.}} 
\label{fig:boundary_grouping} 
\end{figure}
This is not true if $v$ belongs to the surface. Still, let us
determine the contribution of each group. We first note 
that groups of the first type (for which the three mid-edges $p$, $q$, and $r$
are visited) cannot exist when $v$ is on the surface. For groups of the
second type, we distinguish two cases, depending on whether the
walk approaches $v$ via $q$ or via $r$ 
(Fig.~\ref{fig:boundary_grouping}). If the leftmost configuration in
each group of
Fig. \ref{fig:boundary_grouping} is denoted $\gamma_1$, and the
rightmost one  $\gamma$,  with associated SAW components $w_1$ and
$w$, then the contribution in the  first case is
\beq\label{eqn:boundaryeqn_first}
(q-v) \x^{|\gamma_1|}y^{c(\gamma_1)}n^{\ell(\gamma_1)}e^{-\ii
  \sigma W(w_1)}(1+\x y \bar\lambda j + \x y \lambda \bar \jmath )
\eeq
with $j=e^{2i \pi/3}$.
But we know that this vanishes when $y=1$
{(this is Smirnov's result),} 
so the last term in parentheses must be $(1-y)$. Moreover, 
$$
|\gamma_1|=|\gamma|-1, \quad 
c(\gamma_1)= c(\gamma)-1, \quad \ell(\gamma_1)=\ell(\gamma),  
                 \quad {W(w_1)=W(w)-\pi/3,}
$$
 and one concludes that groups of walks 
{approaching $v$ via $q$}
 give the first sum in~\eqref{surface-case}. Similarly, for a group of
 walks 
{approaching $v$ via $r$,} 
the contribution is 
\begin{multline}
(r-v)\x^{|\gamma_1|}y^{c(\gamma_1)}n^{\ell(\gamma_1)} e^{-\ii
  \sigma W(w_1)}
(1+\x  y \bar \jmath \lambda+\x  y  j\bar \lambda) =
\label{eqn:boundaryeqn_second}\\
(r-v)(1-y)
\x^{|\gamma|-1}y^{c(\gamma)-1}n^{\ell(\gamma)} e^{-\ii \sigma  (W(w)+\pi/3)},
\end{multline}
which gives the second sum in~\eqref{surface-case}.
\end{proof}

In \cite{DC-S10}, Duminil-Copin and Smirnov 
{prove and } 
use Lemma~\ref{lem:local} to prove that the growth constant of the
self-avoiding walk is given by the case $n=0$ of  the dilute
regime~\eqref{eq:Slemma_dilute}: 
 $\x^{-1}:=2\cos(\pi/8)=\sqrt{2+\sqrt{2}}$. 
They do so by considering a special trapezoidal domain 
$D_{L,T}$ as shown in Fig.~\ref{fig:S_boundary},
and deriving from the local identity a global identity that relates
several \gfs\ counting walks in this domain. Here we
generalise this identity to a general $O(n)$ model 
including a boundary weight. 

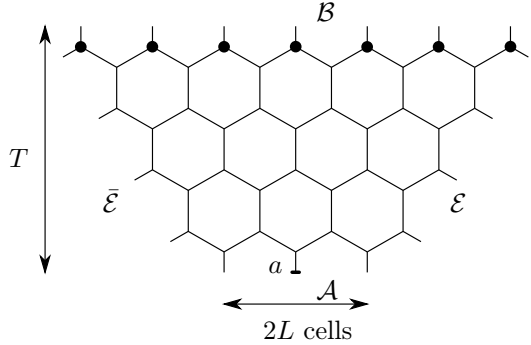
\begin{figure}[h]
\begin{center}
\begin{picture}(200,120)
\put(0,0){\includegraphics[height=120pt, angle=0]{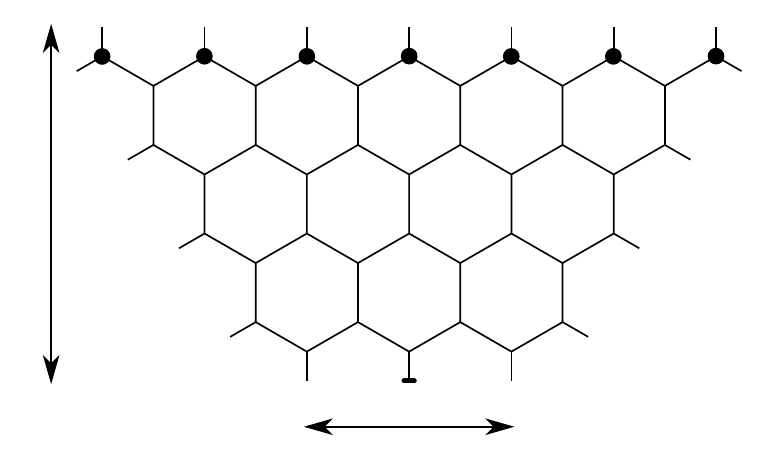}}
\put(115,10){$\cA$}
\put(115,115){$\cB$}
\put(165,43){$\cE$}\put(35,43){$\bcE$}
\put(97,20){$a$}
\put(95,-5){$2L$ cells}
\put(0,60){$T$}
\end{picture}
\end{center}
\caption{Finite patch $D_{T,L}$  of the half hexagonal lattice, with
  $T=4$ and $L=1$ {(the
  convention on $T$ is chosen in such a way a walk  of minimal
  length going from the bottom to the top of the domain contains
  $T-1$ vertical edges and two vertical half-edges, one at each end of the walk)}.
The SAW components of 
configurations start  on the central mid-edge $a$ of the bottom boundary.
The weighted vertices, belonging to the surface, are
marked with a black disc.}
\label{fig:S_boundary}
\end{figure}

We partition the boundary $\partial D_{T,L}$ into four subsets
$\cA$, $\cB$, $ \bcE$ and $ \cE$ as illustrated in Fig.~\ref{fig:S_boundary}. We
also  define four  generating functions, counting configurations
in $D_{T,L}$ starting from $a$ and ending in $\partial
D_{T,L}$. First,
\begin{align}
A_{T,L}(x,y) &:= \sum_{\gamma: a \leadsto \cA\setminus\{a\}}
 x^{|\gamma|}y^{c(\gamma)}n^{\ell(\gamma)}, 
\end{align}
where the sum is over all configurations in $D_{T,L}$ whose SAW
component goes from the mid-edge $a$ to a mid-edge of
$\cA\setminus\{a\}$. We  similarly define the
\gfs\ $A^\circ_{T,L}(x,y)$, $B_{T,L}(x,y)$ and $E_{T,L}(x,y)$ for
configurations ending in
$\{a\}$, $\cB$, and $ \bcE\cup \cE$ respectively.
Note that  configurations counted by $A^\circ$  comprise \emph{only}
closed loops inside $D_{T,L}$; that is, their SAW
component is the empty walk $a\leadsto a$. 

\begin{Proposition}
\label{prop:global}
For $n=2\cos \theta$ and $\x^{-1}:= 2\cos ((\pi\pm\theta)/4)$, the
above defined \gfs\ satisfy
\beq \label{eqn:loop_invariant}
A^\circ_{T,L}(\x,y)=\cos \left(\frac{3(\pi\pm \theta)}{4}\right) A_{T,L}(\x ,y) + \cos \left(\frac{\pi\pm\theta}{2}\right) E_{T,L}(\x ,y) + \frac{\ys-y}{y(\ys-1)} B_{T,L}(\x ,y),
\eeq
where
\[
\ys = \frac{1}{1-2\x^2} = 1 \mp \frac2{\sqrt{2-n}}.
\]
\end{Proposition}
Observe that in the dilute case  $\x^{-1}= 2\cos ((\pi-\theta)/4)$,
the value of $\ys$ coincides with the predicted value of $\y (n)$
given in Conjecture~\ref{conj1}. In Section~\ref{sec:proof}, we use the above
identity to prove Conjecture~\ref{conj1} in the case $n=0$
{(that is, Theorem~\ref{thm:THEOREM}).} 
In this case the left-hand
side of~\eqref{eqn:loop_invariant} reduces to $1$, all
coefficients are positive as long as
$y<\ys$,  so that the polynomials $A_{T,L}$, $B_{T,L}$ and  $E_{T,L}$ are
uniformly bounded, independently of $T$ and $L$. Just as in the proof of
Duminil-Copin and Smirnov for the growth constant of  SAWs,  the bound on
$B_{T,L}$ is an important ingredient of our proof. The
identity~\eqref{eqn:loop_invariant} allows $B_{T,L}(\x ,y)$ to diverge
for $y\geq \ys$  
{(as $T$ and $L$ grow)}
which signals the surface transition at the $\cB$ boundary. 
\begin{proof}
Let $p_v, q_v, r_v$ be the mid-edges adjacent to a vertex $v$. 
{Let $F$ be the observable defined by~\eqref{F-def}, and take $\sigma= (\pi
\mp 3\theta)/(4\pi)$ as in Lemma~\ref{lem:local}.}
We compute the sum
\beq\label{eqn:sumFoverS}
S:=\sum_{v\in V(D_{T,L})}\big((p_v-v)F(p_v)+(q_v-v)F(q_v)+(r_v-v)F(r_v)\big)
\eeq
in two ways. 

Firstly, all summands of~\eqref{eqn:sumFoverS} associated with a
non-weighted vertex $v$ are 0 by the first part of
Lemma~\ref{lem:local}. We are left with the contribution of vertices
lying on the surface, given in the second part of the lemma. Since
$W(w)=0$ for all walks occurring in~\eqref{surface-case},
\begin{multline*}
2S=e^{-5i\pi/6}(1-y)(\x y\lambda)^{-1}\sum_{p \in \cB, \gamma : a
  \leadsto q, p} \x ^{|\gamma|} y^{c( 
  \gamma)} n^{\ell(\gamma)} \\
+e^{-i\pi/6}(1-y) (\x y\bar \lambda)^{-1}\sum_{p \in \cB, \gamma :
  a\leadsto r, p} \x ^{|\gamma|} y^{c( \gamma)} n^{\ell(\gamma)}, 
\end{multline*}
where $q$ (resp. $r$) stands for the SW (resp. SE) mid-edge adjacent to $v$.
The factor 2 accounts for the fact that edges have  length $1$, so
that terms like $(p-v)$ have modulus $1/2$. 
Now reflecting a
configuration $\gamma$ that reaches  a mid-edge $p\in \cB$ from the SW
gives a configuration $\gamma'$ that reaches a mid-edge $p'\in \cB$ from
the SE. Moreover, $|\gamma|=|\gamma'|$, $c(\gamma)=c(\gamma ')$ and
$\ell(\gamma)= \ell(\gamma')$. Hence
\begin{eqnarray}
  2S&=&(1-y) (\x y)^{-1}\sum_{p \in \cB, \gamma : a
  \leadsto q, p}\x ^{|\gamma|} y^{c( 
  \gamma)} n^{\ell(\gamma)}
\left( e^{-5i\pi/6} \bar \lambda +e^{-i\pi/6}\lambda
\right)\nonumber\\
&=& -2\ii (1-y) (\x y)^{-1}\cos \left( \frac{\pi \pm \theta}4\right) \sum_{p \in \cB, \gamma : a
  \leadsto q, p}\x ^{|\gamma|} y^{c( 
  \gamma)} n^{\ell(\gamma)} \nonumber
\\
&=& -\ii (1-y) (\x y)^{-1}\cos \left( \frac{\pi \pm
    \theta}4\right)B_{T,L}(\x,y)   \hskip 10 mm \hbox{by symmetry}\nonumber
\\
&=& -\frac \ii 2 (1-y) (\x ^2 y)^{-1}B_{T,L}(\x,y).
\label{eqn:boundary_contributions} 
\end{eqnarray}

To obtain another expression for $S$, starting
from~\eqref{eqn:sumFoverS},  note that any 
mid-edge $p$ not belonging to $\partial D_{T,L}$
 contributes to two terms in the sum,
for vertices $v_1$ and $v_2$, and these two terms  cancel 
because $(p-v_1)=-(p-v_2)$. 
Thus we are left with precisely the contributions
of those mid-edges in $\partial D_{T,L}$: 
\beq\label{eqn:boundary_contributions-bis}
2S=-\ii\sum_{p\in\cA} F(p) + e^{-5\ii\pi/6}  \sum_{p\in\bcE}F(p) + 
e^{-\ii\pi/6}  \sum_{p\in\cE}F(p) + \ii\sum_{p\in\cB}F(p).
\eeq
We again use  symmetry arguments to rewrite this sum. 
First, denoting
$\cA=\{a\} \cup \cA^- \cup \cA^+$ (with $\cA^-$ to the left of $a$), we have
\begin{eqnarray*}
 \sum_{p\in\cA} F(p)&=&A^\circ_{T,L}(\x ,y)+ \sum_{\gamma: a\leadsto \cA^-}
\x^{|\gamma|} y^{c(\gamma)} n^{\ell(\gamma)} \left( \lambda^3 +\bar
  \lambda ^3\right)\\
&=&
A^\circ_{T,L}(\x ,y)- \cos \left( \frac{3(\pi \pm \theta)}4\right)A_{T,L}(\x,y).
\end{eqnarray*}
Similarly,
\begin{eqnarray*}
e^{-\ii\pi/3}\sum_{p\in\bcE}F(p) + e^{\ii\pi/3}  \sum_{p\in\cE}F(p)
&=&\sum_{\gamma: a\leadsto \bcE}\x^{|\gamma|} y^{c(\gamma)}
n^{\ell(\gamma)}
\left(e^{-\ii\pi/3} \lambda^2 +e^{\ii\pi/3} \bar   \lambda ^2\right)\\
&=&
-\cos \left( \frac{\pi \pm \theta}2\right)E_{T,L}(\x,y).
\end{eqnarray*}
Finally,
$$
\sum_{p\in\cB}F(p)= B_{T,L}(\x,y).
$$

Equating \eqref{eqn:boundary_contributions} and~\eqref{eqn:boundary_contributions-bis} gives the proposition.
\end{proof}

\section{Confined self-avoiding walks}
\label{sec:strip}

In the remainder of this paper we specialise to $n=0$,
corresponding to self-avoiding walks.
In this case, 
{we will prove that the critical surface fugacity is $\y=1+\sqrt 2$.}
In this section we first  review some basic but important background,
and then adapt to the honeycomb lattice some known results about
\emm square lattice, SAWs confined to a half-plane or a strip.
{These results will be used in Section~\ref{sec:proof}, where we prove
our main result. }

\medskip
Again, we consider SAWs on the honeycomb lattice, starting and
ending at a mid-edge.
The simplest model associates a fugacity $x$ with each visited vertex 
(or \emm step,, or \emm monomer,). One then studies the generating function 
$$
C(x) = \sum_{k\ge 0} c_k x^k,
$$ 
where $c_k$ is the number of SAWs of $k$ monomers, considered
equivalent up to a translation. A simple concatenation argument and
a classical lemma on sub-multiplicative sequences
suffice to prove that the \emm growth constant,
$$
 \mu:=\lim_{k \to \infty} \left( c_k\right)^{1/k}
$$
 exists and is finite~\cite[Chap.~1]{MadrasSlade93}.
Of course, $1/\mu$ is the radius of convergence $\x$ of the series
$ C(x)$. 
Duminil-Copin and Smirnov~\cite{DC-S10} proved Nienhuis's
conjecture \cite{N82} that, for the honeycomb
lattice, $\mu = \sqrt{2 + \sqrt{2}}.$ 

\subsection{Self-avoiding walks in a half-plane}
\label{sec:hp}
We now consider SAWs in the upper half-plane,
originating at a mid-edge $a$ just below  the surface
(Fig.~\ref{fig:half-space}).  It is known that the growth constant for
such walks is the same as for the bulk case~(see \cite{W75} 
or~\cite[Chap.~3]{MadrasSlade93}). 
We also add a fugacity 
$y$ to vertices in the surface. In physics terms, $y =
\e^{-\epsilon/k_BT}$  where $\epsilon$ is
the energy associated with a surface vertex, $T$ is the absolute
temperature and $k_B$ is Boltzmann's constant. 

 \begin{figure}
\centering
\includegraphics[scale=0.4]{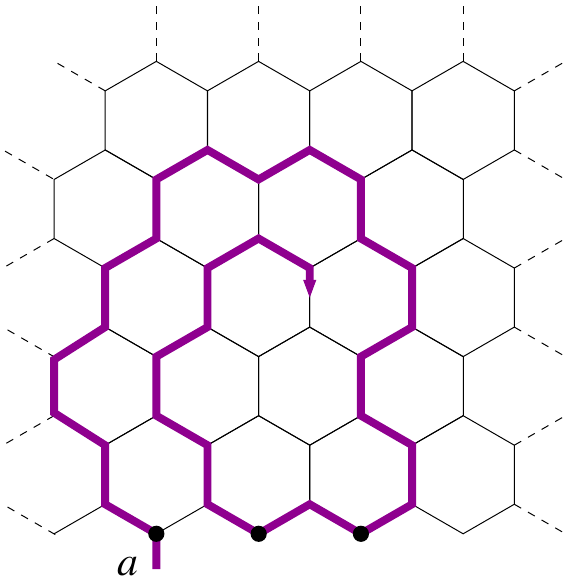}
\caption{A \saw\ in a half-plane, with weights attached to the
  vertices of the surface (indicated by black discs).}
\label{fig:half-space}
\end{figure}

We now consider the following partition function  (which is a
polynomial in $y$): 
\[
{c_k^+(y) := \sum_{|w|=k} y^{c(w)}, }
\]
where the sum runs over half-plane SAWs $w$ of length $k$ and $c(w)$ denotes the
number of \emm contacts, of $w$ with the surface (\emm i.e.,, the number of
vertices of the surface visited by $w$).
\begin{Proposition}\label{prop:half-space}
For $y>0$, 
$$ 
\mu(y):= \lim_{k  \to \infty}c^+_k(y)^{1/k} 
$$ 
exists and is finite. It is a log-convex, 
non-decreasing function of $\log y$, and therefore continuous
and almost everywhere differentiable.

For $0<y\le 1$, 
$$
\mu(y)=\mu(1)= \mu.
$$
Moreover, for any $y>0$,
$$
\mu(y) \ge \max (\mu,\sqrt y ).
$$
This behaviour implies the existence of a critical value $\y$, with 
$1\le \y \le \mu^2$, such that
$$
\mu(y) \left\{
  \begin{array}{ll}
    = \mu & \hbox{ if } y\le \y ,\\
>\mu & \hbox{ if } y >\y .
  \end{array}
\right.
$$
\end{Proposition}
\begin{proof}
 The existence of $\mu(y)$ has been
proved by  Hammersley, Torrie and Whittington~\cite{HTW82} in the case
of the $d$-dimensional hypercubic lattice.
Their discussion and proof,
which use concatenation and \emm unfolding, of walks,
apply {\em mutatis mutandis} to the honeycomb lattice. 
Unfolding consists of reflecting parts of the walk in 
{vertical lines}
passing through those
vertices of the walk with maximal and minimal $x$-coordinates (Fig.~\ref{fig:unfolding}).  This
unfolding is repeated until the origin and end-point have 
 minimal and maximal $x$-coordinates respectively. 
The main advantage of such unfolded walks is that they can be concatenated
without creating self-intersections 
(this may require the   addition of a few steps between the walks).

\begin{figure}
\centering
\includegraphics[scale=0.4]{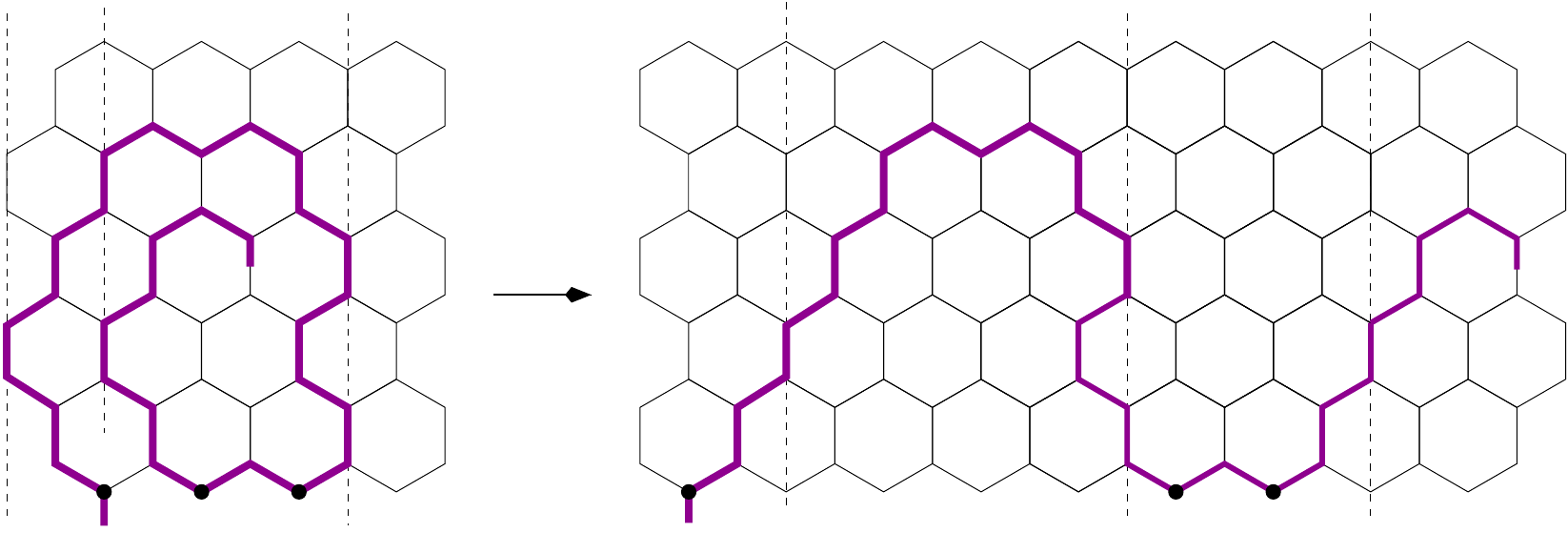}
\caption{Unfolding a half-plane SAW on the honeycomb lattice.}
\label{fig:unfolding}
\end{figure}

The other results are elementary, and adapted from an earlier paper
of Whittington~\cite{W75}. In particular, the lower bound
$\mu(y)\ge \sqrt y$
is obtained by counting zig-zag walks sticking to the surface.
\end{proof}

{The function $\mu(y)$ is not known explicitly, but the main result of
this paper is that $\y=1+\sqrt 2 = \mu ^2-1$.
The behaviour of $\mu(y)$ as $y \to \infty$ has  recently been 
established by Rychlewski and Whittington~\cite{RW11},} who proved
that, on the square lattice, $ \mu(y)$ is asymptotic to $y$. This
translates into $\mu(y)\sim \sqrt y$ in our honeycomb setting.

\medskip

The critical value $\y$, which we have  defined in analytic terms, can
also be given a probabilistic description. Fix $k$, and
assign to each half-plane SAW $w$ of length $k$ the probability
$$
\frac{y^{c(w)}}{c_k^+(y)}.
$$
If $y$ is large, this probability distribution favours walks with many
contacts, 
 while if $y$ is small, the walk is repelled by the surface.
 The mean density of vertices of the walk lying in the surface is 
$$
\frac 1 {k\,  c_k^+(y)} \sum_{|w|=k} c(w) y^{c(w)}=
\frac{y}{k}\frac{\partial \log c^+_k(y)}
{\partial  y}.
$$ 
Recall that $\frac 1 k \log c^+_k(y)$ tends to $\log \mu(y)$
as  $k \rightarrow \infty$.
In the limit of infinitely long walks, 
it can be shown\footnote{The exchange of the limit and the
 derivative is possible thanks to the convexity of $\log \mu(y)$, see
 for instance~\cite[Thm.~B7, p.~345]{MR1858028}.} that the above density tends
to
$$
y\frac{\partial \log \mu(y)}{\partial y}.
$$
 From the behaviour of $\mu(y)$ given in Proposition~\ref{prop:half-space},
we see that the density of vertices on
the surface is $0$ for $y < \y$ and is positive for $y > \y.$  
In other words, the critical value $\y$ 
{distinguishes between}
 {the {\em    desorbed} and {\em     adsorbed} phases.

\subsection{Self-avoiding walks in a strip}
\label{sec:interacting_bridges}

As discussed in the previous subsection,
the usual model of surface-interacting {polymers}
  considers walks
originating in a surface and interacting with monomers (or edges) in
that surface. One way to study such systems is to consider interacting
walks in a strip, and then to take the limit as the strip width
becomes infinite. Clearly, if one studies walks in a strip, it is
possible to consider interactions with both the top and bottom
surface. 
{The results of this section will reconcile the
(apparently inconsistent) settings of Section~\ref{sec:identity} (where
weighted vertices are at the top of the domain, opposite the
starting point) and Section~\ref{sec:hp} (where weighted vertices lie at the
bottom of the domain, on the same side as the starting point).}

Consider a strip of height  $T$ on the honeycomb lattice, as shown in
Fig.~\ref{fig:bridge}.  
We consider SAWs that originate at a mid-edge $a$ just below the bottom of the
strip. Such  walks are said to be  \emm arches, if they end at the
bottom of the strip, and  \emm bridges, if they end at the top (Fig.~\ref{fig:bridge}).
{We now consider the bivariate polynomials
\[
c_{T,k}(y,z)
= \sum_{|w|=k} y^{bc(w)}z^{tc(w)},
\]
where the sum runs over all SAWs $w$ of length $k$ in the $T$-strip and
$bc(w)$ and $tc(w)$ are the numbers of contacts of $w$
with the bottom and top of the strip respectively.} 
We define similar polynomials $a_{T,k}(y,z)$}and
$b_{T,k}(y,z)$ for arches and bridges.

\begin{figure} 
\centering
\begin{picture}(400,110)
\put(0,0){\includegraphics[height=100pt]{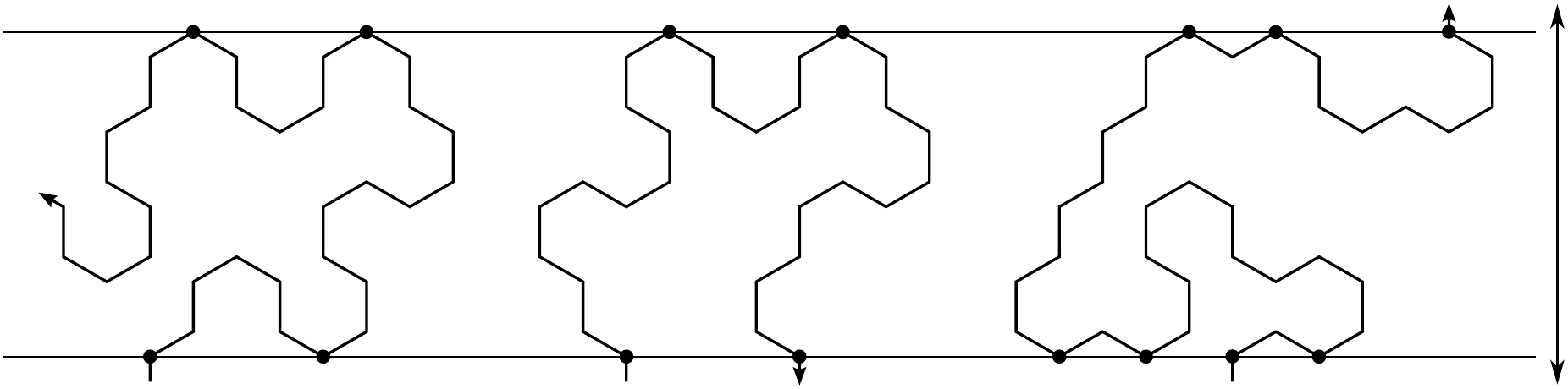}}
\put(405,44){$T$}
\put(35,-4){$a$}
\end{picture}
\caption{Walks confined to  a strip of height $T=5$ with weights
  attached to vertices along the top and 
bottom of the strip: a general walk, an arch, and a bridge.}
\label{fig:bridge}
\end{figure}

\begin{Proposition}\label{prop:strip}
  For $y, z>0$, one has
$$
\lim_{k \to \infty}  a_{T,k}(y,z)^{1/k}= \lim_{k \to \infty} b_{T,k}(y,z)^{1/k} =\lim_{k \to \infty}  c_{T,k}(y,z)^{1/k}:=\mu_T(y,z),
$$
where $\mu_T(y,z)$ is finite, and non-decreasing in $y$ and $z$. By the symmetry of bridges, 
 $$
\mu_T(y,z)= \mu_T(z,y),
$$
 and so, in particular, $\mu_T(y,1)= \mu_T(1,y).$
Finally, $\mu_T(1,y)$ is a log-convex and thus continuous
function of $\log(y)$.
\end{Proposition}
\begin{proof}
Again,  the existence of the limits
follows from concatenation and unfolding arguments as given in
Section~4 of~\cite{JOW06}. 
The log-convexity result is easily adapted from~\cite[Thm.~6.3]{JOW06}.
\end{proof}

Therefore the {growth constant}
 for interacting SAWs
 in a strip is independent of which wall the interacting monomers are
situated on. As per our discussion in  Section~\ref{sec:identity}, it turns
out to be convenient to put the interacting monomers on the 
top, rather than at the bottom.

The next proposition describes how the growth constant $\mu_T(1,y)$
changes as $T$ grows.
\begin{Proposition}\label{prop:strip-convergence}
  For $y>0$, we have 
$$
\mu_T(1,y) <\mu_{T+1}(1,y).
$$
Moreover, as $T\rightarrow \infty$,
$$
\mu_T(1,y)\rightarrow \mu(y),
$$
the growth constant of SAWs interacting with a surface (Proposition~{\rm\ref{prop:half-space}}).
\end{Proposition}
\begin{proof}
Again, the proof is an adaptation to the honeycomb lattice of results
proved by van Rensburg, Orlandini and Whittington for the hypercubic
lattices~\cite{JOW06} (similar 
arguments are also covered in Chapter~8
of~\cite{MadrasSlade93}, but without interactions). Our arguments are
similar to Sections~5 and~6 of~\cite{JOW06}, but, we believe, somewhat shorter\footnote{In particular, working in
  two dimensions gives a simple argument proving the divergence  at
  their radius of convergence of \gfs\ that count SAWs in a strip.
  Moreover, we do not need the full strength of a pattern theorem.}. 

First, since $\mu_T(1,y)=\mu_T(y,1)$, we may choose to work with arches in
a strip of height $T$, interacting with the bottom line of the
strip. Let us say that an arch going from mid-edge $a$ to mid-edge $b$
is \emm unfolded, if the abscissa {$\xco(v)$} of every non-final vertex $v$ of the walk
satisfies {$\xco(a)\le \xco(v)<\xco(b)$}. {That is, an arch is
  unfolded if its starting point is (not necessarily strictly) to the
  left of all other points, and its final point is strictly to the
  right of all other points.} Two unfolded arches $w_1$ and $w_2$ can
be concatenated 
{(after deleting the last half-edge of the first arch and the first
half-edge of the second arch}, see
Fig.~\ref{fig:concatenation}) to form a new unfolded arch $w$. 
{Observe that
$$
|w|=|w_1|+|w_2|-1 \quad \hbox{and} \quad c(w)=c(w_1)+c(w_2)-1.
$$}
We say an unfolded arch is \emm prime, if it is not the concatenation
of two (or more) unfolded arches. The first two arches of
Fig.~\ref{fig:concatenation} are prime, the third one, by
construction, is not. 

\begin{figure} 
\centering
{\includegraphics[height=80pt]{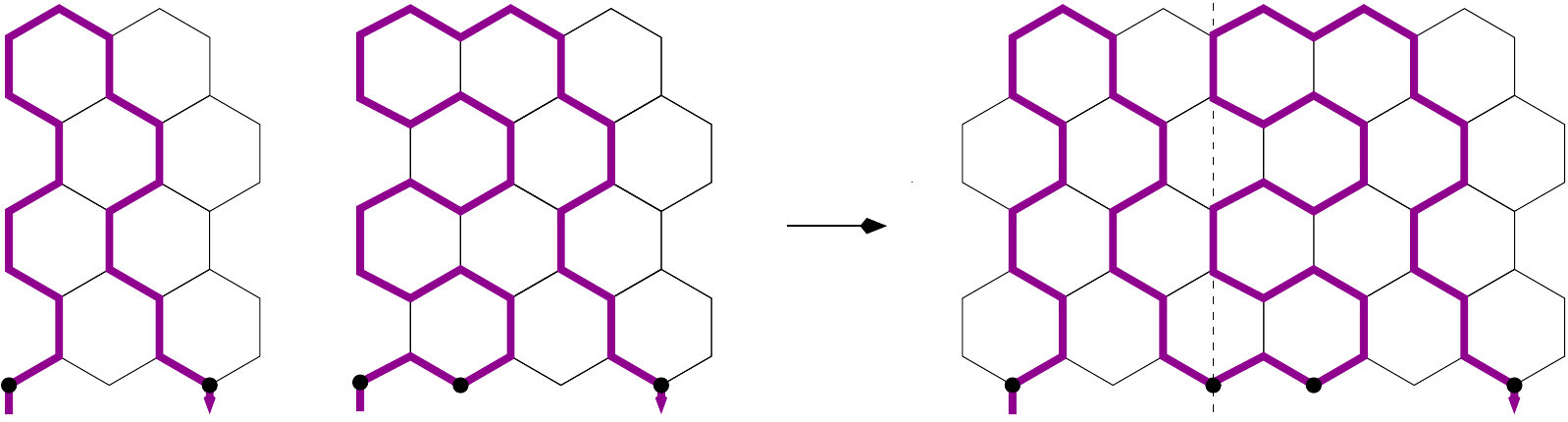}}
\caption{Concatenation of two unfolded arches in a strip of height $T=5$.}
\label{fig:concatenation}
\end{figure}

Let us fix $y>0$. The arguments of~\cite[Section~4]{JOW06} show that the \gf \
$\vec A_T(x,y)$ that counts unfolded arches (by the size and the number of
contacts with the bottom line of the strip) has the same radius of
convergence as the \gf \
$A_T(x,y)$ that counts all arches. By Proposition~\ref{prop:strip},
this radius is $\rho_T(y):=1/\mu_T(y)$. Moreover, the above
definition of prime arches shows that 
$$
\vec A_T(x,y)= \frac{P_T(x,y)}{1-P_T(x,y)/(xy)},
$$
where $P_T(x,y)$ counts prime unfolded arches. 

It follows from the transfer matrix method that the series $\vec
A_T(x,y)$ (and, in fact, all series counting walks in a strip that
occur in this section) is  a rational function of $x$ and $y$
(see~\cite[p.~364]{FlajSedg}, or~\cite{alm-janson}). 
Hence $\vec A_T(x,y)$ diverges at its radius $\rho_T(y)$,
and it follows that $P_T(\rho_T(y),y)/(y\rho_T(y))=1$. 

Now consider the prime unfolded arch $w$ that consists of a (wavy) column
with $2(T-1)$ vertical edges (like the first arch of Fig.~\ref{fig:concatenation}). This walk contributes a term
$x^{4T-1}y^2$ in the series $P_T(x,y)$. Let $\tilde P_T(x,y):=
P_T(x,y)- x^{4T-1}y^2$. The \gf\ of unfolded arches that do not
contain $w$ as a factor is
$$
\frac{\tilde P_T(x,y)}{1-\tilde P_T(x,y)/(xy)}.
$$
{Its radius is reached at the point $x$ satisfying $\tilde
 P_T(x,y)/(xy)=1$.} 
It is hence larger than the radius $\rho_T(y)$ of $\vec
A_T(x,y)$. The above series counts (among others)  walks that do not
touch the top line of the strip.  Their  \gf\  is  $\vec
A_{T-1}(x,y)$, which has radius $\rho_{T-1}(y)$. Hence
$\rho_{T-1}(y)>\rho_T(y)$, or equivalently $\mu_{T-1}(y)<\mu_T(y)$. 

\medskip
The proof that $\mu_T(y)$ tends to $\mu(y)$ is analogous to the proof
of Theorem~6.5 in~\cite{JOW06}.

\end{proof}

We now derive a corollary that will be essential in the
next section. It deals with the properties of
$\rho_T(y):=1/\mu_T(1,y)$, which is the 
radius of convergence of the series 
$$
 C_T(x,y):= \sum_{k\ge 0} c_{T,k}(1,y)x^k
$$
counting walks in a strip that interact with the top boundary, and of
the analogous series $ A_T(x,y)$ and $ B_T(x,y)$ that count arches and
bridges. See Fig.~\ref{fig:rho} for an illustration.
\begin{figure}
\centering
\scalebox{0.4}{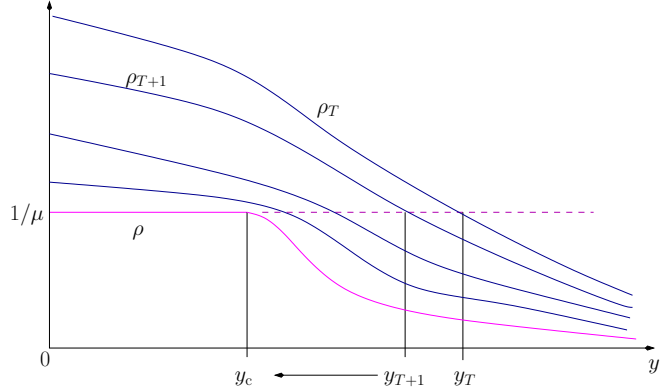}
\caption{An illustration of Corollary~\ref{coro:yT-yc}.} 
\label{fig:rho}
\end{figure}

\begin{Corollary}\label{coro:yT-yc} 
Let $y>0$. The \gfs\ $A_T(x,y)$, $B_T(x,y)$ and
  $C_T(x,y)$ all have the same radius of convergence, 
$$
\rho_T(y)=1/\mu_T(1,y).
$$
Moreover, $\rho_T(y)$ decreases to $\rho(y):=1/\mu(y)$ as $T$ goes to
infinity. In particular, $\rho_T(y)$ decreases to 
$1/\mu$ for $y\le\y $. 

There exists a unique $y_T>0$ such that $\rho_T(y_T)=\x =1/\mu$. The
series (in $y$) $A_T(\x,y)$, $B_T(\x,y)$ and 
  $C_T(\x,y)$ have  radius of convergence
$y_T$, and  $y_T$ decreases to the critical fugacity $\y $ as  $T$
goes to infinity. 
\end{Corollary}
\begin{proof}
  The first part of the {corollary} is an obvious translation of
  Propositions~\ref{prop:strip} and~\ref{prop:strip-convergence}. 

The existence of $y_T$ follows from the intermediate value theorem:
$\rho_T$ is continuous,  $\rho_T(1)>\rho(1)=\x$ and $\rho_T(y)\rightarrow 0$ as $y \to
\infty$ (because $\rho_T(y)\le 1/\sqrt y$ as can be seen by counting
zig-zag paths, as in the proof of Proposition~\ref{prop:half-space}).

The uniqueness of $y_T$ follows from the log-convexity of $\mu_T(y)$
in $\log y$, 
{together with $\rho_T(1)>\x$: }
this precludes having $\rho_T(y)=\rho_T(y') = \x$
{with $y\not = y'$.}
This also means that
\beq\label{equiv}
\rho_T(y)<\rho_T(y_T) \Longleftrightarrow y>y_T 
\qquad \hbox{ and } \qquad 
\rho_T(y)>\rho_T(y_T) \Longleftrightarrow y<y_T .
\eeq

Let us now prove that $y_T$ is the radius {of convergence} of $A_T(\x ,y)$,
$B_T(\x ,y)$ and $C_T(\x ,y)$.   The argument is the same for the
three series, so let us work for instance with $C_T$. 
{We must first explain  why $C_T(\x,y)$ is indeed a series in $y$, that
is, why the length \gf\ of SAWs in the $T$-strip having a fixed number
of top contacts is finite at $\x$. The reason for that is that the
number of $k$-step walks of this type grows (at most) like
$\mu_T(1,1)^k$, by definition of $\mu_T$, and that $\mu_T(1,1)<1/\x$
as stated at the beginning of the corollary. Hence $C_T(\x,y)$ is
indeed a series in $y$.}
Now by definition of $\rho_T$, the series $C_T(\x ,y)$ converges if
$\x <\rho_T(y)$, and diverges if $\x >\rho_T(y)$. But
$\x=\rho_T(y_T)$, so by~\eqref{equiv}, this means that $C_T(\x ,y)$ converges
if $y<y_T$ and diverges if $y>y_T$, which means that $y_T$ is the
radius {of convergence} of $C_T(\x ,y)$.

Let us finally prove that $y_T$ decreases towards $\y $. First, since
$\rho_T(y_T)=\x $ and $\rho_{T+1}(y)<\rho_T(y)$
(Proposition~\ref{prop:strip-convergence}), we have
$\rho_{T+1}(y_T)<\x $ and thus $y_{T+1}<y_T$. Hence the sequence
$(y_T)_{T\ge1}$ 
decreases. Let $\bar y$ be its limit. For $y\le \y $, we
have $\rho_T(y)>\rho(y)=\x $, and thus $y_T>\y $ for all $T$. Hence
$\bar y\ge \y $. Since $\bar y < y_T$, we have $\rho_T(\bar
y)>\rho_T(y_T)=\x $, and thus $\rho(\bar y)\ge \x $
(Proposition~\ref{prop:strip-convergence}). Since $\rho(y)<\x  $ for
$y>\y $ (Proposition~\ref{prop:strip}), it follows that $\bar y \le
\y $. We have thus proved that $y_T$ decreases to $\y $.
\end{proof}

\section{The critical surface fugacity  of SAWs is $1+\sqrt{2}$}
\label{sec:proof}
\subsection{The global identity}
Let us {consider} the identity~\eqref{eqn:loop_invariant} at $n=0$, 
that is, at $\theta=\pi/2$.
Then no loops are allowed. In particular, the polynomial $A^\circ_{T,L} $
reduces to 1. 
{We focus now on the dilute regime~\eqref{eq:Slemma_dilute}, where
$$
\x ^{-1} = 2 \cos \left( \frac{\pi  }8\right)=
\sqrt{2+\sqrt{2}}\qquad \hbox{and } \qquad \ys = 1+\sqrt{2}. 
$$
The global identity reads}
\beq\label{eqn:identity_simplified}
1 = {\alpha} A_{T,L}(\x ,y) + {\eps}E_{T,L}(\x ,y) + {\beta}(y)B_{T,L}(\x ,y),
\eeq
where 
$$
{\alpha} = \cos\left(\frac{3\pi}{8}\right) =
\frac{\sqrt{2-\sqrt{2}}}{2},
\qquad {\eps} = \cos\left(\frac{\pi}{4}\right) = \frac{1}{\sqrt{2}},\qquad
{\beta}(y) = \frac{\ys-y}{y(\ys-1)} = \frac{1+\sqrt{2}-y}{\sqrt{2}\,y}.
$$

\subsection{A lower bound on $\y $}
\label{sec:lower}
Let us fix $T$, and set $y=y^*$
in~\eqref{eqn:identity_simplified}. Since $\beta(y^*)=0$, we  obtain: 
$$
1 = {\alpha} A_{T,L}(\x ,y^*) + {\eps}E_{T,L}(\x ,y^*).
$$
As $L$ increases, the numbers $A_{T,L}(\x ,y^*)$ 
count more and more walks.  Hence  they increase with~$L$.  
Since the  coefficients ${\alpha}$ and ${\eps}$ 
are  positive, the above identity shows that 
$A_{T,L}(\x ,y^*)$ remains bounded as $L$ increases. Hence
the limit
$$
\lim_{L\rightarrow \infty} A_{T,L}(\x ,y^*) 
$$
exists and is finite. Clearly, this limit is $A_T(\x ,y^*)$ 
where $A_T(x,y)$ is
the \gf\ of arches in a $T$-strip, 
defined just above Corollary~\ref{coro:yT-yc}.  
According to this corollary,   $A_T(\x,y)$ has radius $y_T$. Since it
converges at $y^*$, this means that $\ys\le y_T$. Since $y_T
\rightarrow \y$, we thus have 
\beq\label{lb}
y^* \le \y .
\eeq

\subsection{A limit identity}\label{sec:E0}
\begin{Proposition}\label{prop:simplified}
For $0\le y<y_T$ (the radius of convergence of $A_T(\x , \cdot)$ and
$B_T(\x , \cdot)$), the series counting arches and bridges in a
$T$-strip satisfy 
  \beq\label{id2}
\alpha A_{T}(\x ,y)+ \beta(y) B_{T}(\x ,y)=1.
\eeq
\end{Proposition}
\begin{proof}
  Let us first prove  that 
$$
\lim_L E_{T,L}(\x ,y)=0 \quad \hbox{for } 0\le y <y_T.
$$
Indeed, $E_{T,L}(\x ,y)$ counts some \saws \ of length at least $L$,
starting from $a$, and confined to a $T$-strip. But the \gf\ of walks
in the $T$-strip converges at $(\x ,y)$ for $y<y_T$ (see
Corollary~\ref{coro:yT-yc}), and thus its remainder of order $L$ tends to 0 as
$L$ grows. This remainder is an upper bound on $E_{T,L}(\x,y)$, which
thus tends to 0 as well.

Taking the limit of~\eqref{eqn:identity_simplified} as $L\rightarrow
\infty$  gives the proposition. 
\end{proof}

\subsection{Convergence of $B_T(\x ,1)$ to $0$}
\label{sec:bridges-limit}
This is a key point in our argument, and also a result of independent
interest.
\begin{Theorem}\label{thm:BT0}
  The length \gf\ $B_T(x,1)$ counting bridges in a strip of height
  $T$, taken at the critical value $\x = 1/\sqrt{2+\sqrt 2}$, tends
  to $0$ as $T$ tends to infinity.
\end{Theorem}
  The proof, of a probabilistic nature, is given in the appendix. Let us note that the
  fact that $B_T(\x ,1)$ converges (and actually decreases) follows easily
  from the case $y=1$ of~\eqref{id2}. Indeed,
  $A_T(\x ,1)$  increases with $T$, but remains bounded since
 {$\alpha$ and $\beta(1)$} are positive. 
Thus $A_T(\x ,1)$ has a finite limit when
  $T$ increases, and this limit is the \gf\ $A(\x )$ counting arches in a
  half-plane.  It then follows from~\eqref{id2} that $B_T(\x ,1)$
  decreases as $T$ grows, 
and
\beq\label{B-lim}
\lim_T B_T(\x ,1)= 1-\alpha A(\x ).
\eeq
Theorem \ref{thm:BT0} thus  implies  that $A(\x )=1/\alpha$. 

\medskip

\noindent{\bf Remarks}\\
{\bf 1.} We can actually prove that $A_T(\x ,y)
\rightarrow A(\x )$ for $y<\ys$, but this will not be needed here.
Returning to~\eqref{id2}, this implies that $B_T(\x ,y)\rightarrow 0$
for $0\le y <\ys$.

\medskip
\noindent{\bf 2.}
As discussed in~\cite[Remark~2]{DC-S10}, it follows from the SLE predictions
of~\cite[Sec.~3.3.3 and 3.4.3]{LSW5} that
$B_T(\x,1)$ is expected to decay as $T^{-1/4}$ as $T\to \infty$.

\subsection{An upper bound on $\y$}
The series $A_{T+1}(\x ,y)$ counts arches of height at most
$T+1$. This includes arches of height at most $T$, which have no
contacts with the top boundary. 
Such arches are counted by $A_T(\x,1)$.
Now consider an arch that has
contacts with the boundary. By looking at its last contact, one can factor
the arch into two bridges (see Fig.~\ref{fig:arch-bridges}), and thus obtain 
\beq\label{ineq4}
A_{T+1}(\x ,y)-A_T(\x ,1)\le \x  B_T(\x ,1)B_{T+1}(\x ,y).
\eeq
This inequality holds in the domain of convergence of the series it involves, that
is,  for $y<y_{T+1}$, and thus in particular at $\y$.
{Let us now write~\eqref{id2}, first  for $T+1$ and 
$y=\y$ and then for $T$ and $y=1$:
$$
\alpha A_{T+1}(\x ,\y)+ \beta(\y) B_{T+1}(\x ,\y)=1
=
\alpha A_{T}(\x ,1)+  B_{T}(\x ,1).
$$
Combine this with the inequality~\eqref{ineq4}, taken at $y=\y$.
 This gives 
$$
B_T(\x,1)-\beta(\y)B_{T+1}(\x,\y)\le \alpha \x B_T(\x,1)B_{T+1}(\x,\y),
$$
or equivalently,
$$ 
0\le \frac 1{B_{T+1}(\x ,\y)} \le \alpha \x  + \frac 1 {B_T(\x ,1)} \frac
{y^*-\y}{\y(y^*-1)}.
$$ 
Recall that $B_T(\x ,1)$ tends to $0$ (Theorem~\ref{thm:BT0}).
  This forces 
$y^*\ge \y ,$
otherwise the right-hand side would become arbitrarily large in
modulus and negative as $T\to \infty$.}

Together with~\eqref{lb}, this establishes $\y=\ys =1+\sqrt 2$ and
completes the proof of Theorem~\ref{thm:THEOREM}.
\qed

\begin{figure}[htb]
\centering
\scalebox{0.4}{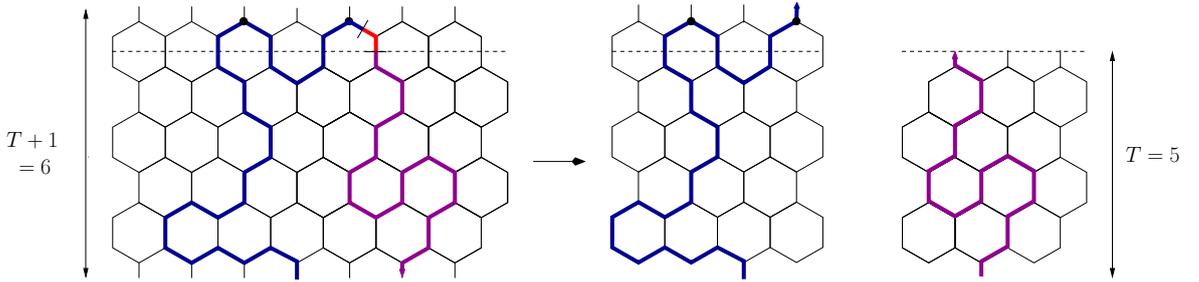}
\caption{Factorisation of an arch of height $T+1$ into two bridges, of
  height $T+1$ and $T$ respectively.}  
\label{fig:arch-bridges}
\end{figure}

\section*{Acknowledgments}
We thank Neal Madras, Andrew Rechnitzer, Stu
Whittington and Alain Yger for helpful conversations.
AJG and JdG acknowledge financial support from the Australian
Research Council. NRB was supported by the ARC Centre of Excellence
for Mathematics and Statistics of Complex Systems (MASCOS).
HDC was supported
by the ANR grant BLAN06-3-134462, the ERC AG CONFRA, as well as by the
Swiss FNS.

Part of this work was carried out during visits of the authors to the
Mathematical Sciences Research Institute in Berkeley, during the
Spring 2012 Random Spatial Processes Program. The authors thank the
institute for its hospitality and the NSF (grant DMS-0932078) for its
financial support. 
%

\spacebreak
\bigskip
\newpage
\begin{center}{\large{\bf Appendix. Proof of Theorem~\ref{thm:BT0}.}}
\end{center}

\medskip

Before starting the proof, let us introduce some additional notation. 
The set of mid-edges of the honeycomb lattice is denoted by $\mathbb
H$. The lattice 
has an origin  $a\in \mathbb H$, at coordinates $(0,0)$. We denote by {$(\xco (v),\yco (v))$} the coordinates of a point
$v\in \cs$ (that is, its real and imaginary parts). 
We consider \saws\ that start and end at a mid-edge. 
A \saw\ $\gamma$ is denoted by the sequence $(\gamma_0,
\ldots, \gamma_n)$ of its mid-edges. The  length of $\gamma$, that
is, the number of vertices of the lattice it visits, is denoted as
before by $|\gamma|=n$. 
{(This $n$ has nothing to do with the $n$ of the $O(n)$ model
considered in Sections~\ref{sec:intro} and~\ref{sec:identity}. We are
dealing in this appendix with SAWs, that is, with the $O(0)$ model.)}
To lighten notation, we often omit floor symbols, especially in
indices: for instance, $\gamma_t$ should be understood as
$\gamma_{\lfloor t \rfloor}$.
The cardinality of a set $A$ is denoted by $|A|$. 

\medbreak

We have so far discussed bridges in a strip of height $T$
(Fig.~\ref{fig:bridge}, right), which we  call bridges \emm of height, $T$. 
In general, we call  \emm bridge, any \saw\ $\gamma=(\gamma_0,
\ldots, \gamma_n)$ that is a bridge
of height $T$ for some $T$. Equivalently, 
{$\yco(\gamma_0)<\yco(\gamma_i)<\yco(\gamma_n)$} for $0< i<n$. 
The set
of bridges of length $n$ is denoted by $\sabset_n$.

 The set $\ner_\gamma$ of {\em renewal
  points}  of $\gamma\in \sabset_n$ is the set of points of the form
$\gamma_i$ with  $0 \leq i \leq n$, for which 
$\gamma_{[0,i]}:=(\gamma_0, \ldots, \gamma_i)$ and
$\gamma_{[i,n]}:=(\gamma_i, \ldots, \gamma_n)$ 
are  bridges. 
We denote by $\mathbf{r}_0(\gamma), \mathbf{r}_1(\gamma), \ldots$  the
indices of the renewal points. That is,
$\mathbf{r}_0(\gamma)=0$ and
$\mathbf{r}_{k+1}(\gamma)=\inf\{j>\mathbf{r}_{k}(\gamma):\gamma_j\in
\ner_\gamma\}$ for each $k$.
When no confusion is possible, we  often denote
$\mathbf{r}_k(\gamma) $ by $\mathbf{r}_k$.

A  bridge $\gamma\in \sabset_n$ is {\em irreducible} if its
only renewal points are $\gamma_0$ and $\gamma_n$.
 Let $\isab$ be
the set of irreducible  bridges of arbitrary length 
starting from $a$. Every  bridge $\gamma$
is the concatenation of a finite number
of irreducible bridges,
the decomposition is unique and the set
$\ner_\gamma$ is the union of the initial and terminal points of the
bridges that comprise this decomposition. 

Kesten's relation for irreducible bridges 
(see \cite[Section~4.2]{MadrasSlade93} or \cite{kestenone}) on the
hypercubic lattice $\mathbb Z^d$ can be easily adapted to the honeycomb lattice.
It gives
$$
\sum_{\gamma\in \isab}\x ^{|\gamma|}=1.
$$
This enables us to define a probability measure $\prne$ on
$\mathrm{iSAB}$ by setting  
$\prne(\gamma) = \x ^{|\gamma|}$. Let $\prneinf$ denote the law on
semi-infinite walks $\gamma:\ns \to \mathbb H$ formed by the
concatenation of infinitely many independent samples 
$\gamma^{[1]},\gamma^{[2]},\ldots$
 of $\prne$. 
We refer to~\cite[Section 8.3]{MadrasSlade93} for details of related
measures in the case of $\mathbb Z^d$. The definition of  $\ner_\gamma$ and the
{indexing} of renewal points extend to this context (we obtain
an infinite sequence $(\mathbf r_k
)_{k\in \mathbb N}$). 

{Note that a bridge of length 2 has height 1 and ends at ordinate
$3/2$ (since edges have unit length). More generally, a 
bridge $\gamma$ of length $n$ has height 
{$\height (\gamma)=\frac 2 3 \yco(\gamma_n)$}.
We define  the height of a general SAW $\gamma$ similarly.}
The \emm width, of $\gamma$ is defined by
$$
\width(\gamma)=\frac 1 {\sqrt 3}\max\{\xco(\gamma_k)-\xco(\gamma_{k'}),0\le k,k'\le n\},
$$
so that  a bridge of length 2 has width $1/2$.

We have proved in Section~\ref{sec:bridges-limit} that $B_T(\x ,1)$
converges as $T\to 
\infty$. We provide here an alternative proof, and relate the limiting
value to the average height of irreducible bridges. 
\begin{Lemma}\label{ren sar}
As $T\rightarrow \infty$,
$$
 B_T(\x ,1) \rightarrow
\frac 1 {\mathbb E_{\isab}(\height(\gamma))}.
$$
\end{Lemma}
\begin{proof} 
The result  follows from standard renewal theory. We can for instance
apply~\cite[Theorem 4.2.2(b)]{MadrasSlade93} to the sequence
$$
f_T:= \sum_{\gamma \in \isab:\ \height(\gamma)=T} \x ^{|\gamma|}.
$$
Indeed, with the notation of this theorem, $v_T=B_T(\x ,1)$ and
$\sum_k kf_k= \mathbb E_{\isab}(\height(\gamma))$.
\end{proof}

Thus Theorem~\ref{thm:BT0} is equivalent to
$$
\mathbb E_{\isab}(\height(\gamma))=\infty.
$$ 
We will prove this by contradiction. 
Assuming
$\mathbb E_{\isab}(\height(\gamma))$ 
is finite,  we 
first show that $\mathbb E_{\isab}(\width(\gamma))$ 
is also finite. Then, we
show that under these two conditions, an infinite bridge is very
narrow. The last step 
 consists  in proving that this cannot be the case. The argument uses
 a {\em stickbreak} operation which perturbs a bridge by selecting a 
subpath and rotating it clockwise by $\frac\pi3$. 
The new path is a self-avoiding bridge
for an adequately chosen subpath.
But   its width is relatively large, contradicting the fact that
bridges are narrow.
The strategy of proof is greatly inspired by a recent paper of
Duminil-Copin and Hammond, where \saws\ are  
proved to be sub-ballistic~\cite{DC-H}.
The additional difficulty here comes
from the fact that Section~4 of~\cite{DC-H} (which corresponds to the
proof presented 
here) relies on the assumption $\mathbb E_{\isab}(|\gamma|)<\infty$, which
is stronger than the assumption  $\mathbb
E_{\isab}(\height(\gamma))<\infty$ that we have here. 
In particular, we need the following result.

\begin{Proposition}\label{height width}
If $\mathbb E_{\mathrm{iSAB}}(\height(\gamma))<\infty$, then $\mathbb
E_{\mathrm{iSAB}}(\width(\gamma))<\infty$. 
\end{Proposition}
To prove this proposition we will first establish some simple lemmas regarding SAW generating functions in a slightly different geometry to that used in Section~\ref{sec:identity}.

Consider the rectangular domain $R_{T,L}$ 
depicted in Fig.~\ref{fig:rectangle}, 
with its boundary partitioned into four subsets $\cA$, $\cB$, $\cE^-$
and $\cE^+$ 
(the mid-edges of $\cE^+$ point up, those of $\cE^-$ point down,
{on both sides of the rectangle}). We
do not consider any  interactions here. 
As in Section~\ref{sec:identity}, we define four \gfs\ counting \saws\ in the
rectangle, going from $a$ to a mid-edge of the boundary. First, we set
$$
\tilde A_{T,L}(x) := \sum_
{\gamma: a \leadsto \cA\setminus\{a\}} x^{|\gamma|},
$$
and then the \gfs\ $\tilde B_{T,L}(x)$, $\tilde E_{T,L}^-(x)$ and
$\tilde E_{T,L}^+(x)$ are defined similarly.
We then have the following lemma, which can be viewed as a translation
of \eqref{eqn:identity_simplified} to the new rectangular domain
$R_{T,L}$. We omit the proof, as it is essentially the same as that of
Proposition~\ref{prop:global} with $n=0$ and $y=1$ in the dilute
regime. 
\begin{Lemma}\label{global-rectangle}
The generating functions $\tilde A_{T,L}$, $\tilde B_{T,L}$, $\tilde E^+_{T,L}$ and $\tilde E^-_{T,L}$, evaluated at $x=\x$, satisfy the identity
$$
1=\alpha \tilde A_{T,L}(\x)+\tilde B_{T,L}(\x)+\eps^+\, \tilde
E^+_{T,L}(\x)+\eps^-\, \tilde E^-_{T,L}(\x),
$$
where, as before, 
$\alpha=\cos (\frac{3\pi}8)$,
and now $\eps^-=\cos(\frac\pi 4)$ and $\eps^+=\cos(\frac\pi 8)$.
\end{Lemma}

\begin{figure}[htb]
\centering
\scalebox{0.4}{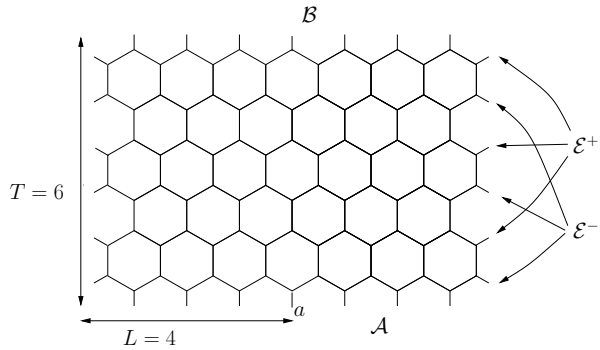}
\caption{The rectangular domain $R_{T,L}$ with $T=6$ and $L=4$.}
\label{fig:rectangle}
\end{figure}

\medskip
 \begin{quote}
   {\bf Convention.} Since we always evaluate our \gfs\ at $x=\x $, we
   will almost systematically omit the variable $\x $, so that $ \tilde A_{T,L}$ now means $   \tilde A_{T,L}(\x )$, and so on.
 \end{quote}
\medskip

We now wish to take the size of the rectangle $R_{T,L}$ to
  infinity to obtain a half-plane, as we did in the previous
  geometry. This time, however, we want  $T$ and $L$ to increase
  together, so the situation here is a bit more delicate.

Recall that an \emm arch,
  is a \saw\ starting from $a$, confined to the upper half-plane, and
  ending on the line {$\yco=0$}. As proved in
  Section~\ref{sec:bridges-limit}, the   \gf\ $A(x)$ of arches
  converges at $\x$. Hence, if $\frak a_L(x)$   denotes the generating
  function of arches ending $L$ cells to the   right of the initial
  mid-edge $a$, the sequence $\frak a_L\equiv   \frak a_L(\x)$ tends
  to $0$ as $L\rightarrow \infty$. 

\begin{Lemma}\label{lem:rectangle_toinfinity_TL}
Assume that  $T\equiv T_k$ and $L\equiv L_k$ 
both tend to infinity as $k$ grows, and that $T\frak a_{2L}\longrightarrow
0$. If $\mathbb E_{\mathrm{iSAB}}(\height(\gamma))<\infty$, then 
\[\lim_{k\to\infty} \tilde B_{T_k,L_k} > 0.\]
\end{Lemma}

\begin{proof}
We begin by bounding $\tilde E^\pm_{T,L}$ in terms of $\frak
  a_{2L}$. For $m\in \mathbb N$, let $\frak e _m^+(x)$ be the
  generating function of walks in $R_{T,L}$ ending on the  right side
  of the rectangle, on the $m$th row of $\cE^+$, so that, by
  symmetry, $\tilde E_{T,L}^+= 2\sum_{m\le \lfloor \frac    T2
    \rfloor} \frak e _m^+$.  
{The  Cauchy-Schwarz inequality gives
$$
\big(\tilde E_{T,L}^+\big)^2 \le 4 \lfloor
  \textstyle
\frac T2\rfloor 
\displaystyle
\sum_{m\le \lfloor \frac
    T2 \rfloor} (\frak e_m^+)^2 .
$$
Now one can concatenate two
walks contributing to $\frak e_m^+$  (after reflecting the second one)
by adding a step between them in order to create an arch  contributing to $\frak
a_{2L}$. This gives}
$$
\big(\tilde E_{T,L}^+\big)^2\le 4 \lfloor \textstyle\frac T2\rfloor\
  \x ^{-1} \frak a_{2L}.
$$
We obtain 
a similar upper bound for $\tilde E_{T,L}^-$
with $\lfloor \textstyle\frac T2\rfloor$ replaced by $\lceil
\textstyle\frac T2\rceil$.

{The assumptions of the lemma now imply that}
 $\tilde E_{T,L}^+$ and $\tilde E_{T,L}^-$ tend to 0. 
Moreover, $\tilde A_{T,L}$ increases with $L$ and $T$, and converges
to $A\equiv A(\x )$.
 Returning to~Lemma~\ref{global-rectangle}
shows that $\tilde B_{T,L}$ must also
converge, and gives
\begin{eqnarray*}
  \lim _{k}\tilde B_{T_k,L_k}&=& 1- \alpha A(\x )\nonumber\\
&=& \lim_T B_T(\x ,1) \quad \quad \hbox{ by } \eqref{B-lim}\nonumber\\
&>&0 \quad \hbox{ by assumption.}
\end{eqnarray*}
\end{proof}

\begin{proof}[Proof of Proposition \ref{height width}]
Let us now return to random infinite bridges and use them to give an
upper bound on $\tilde B_{T,L}$.
Let $0<\delta<1/\mathbb E_{\mathrm{iSAB}}(\height(\gamma))$. We have
\begin{align*}
\tilde B_{T,L}
&= \sum_{\gamma: a \leadsto \cB} \x^{|\gamma|}
\\
&\le\prneinf \big( \exists n\in \mathbb N: 
\height(\gamma_{[0,\mathbf r_n]})=T
\text{ and }\width(\gamma_{[0,\mathbf r_n]})\le 2L\big)\\
&\le \prneinf\big(
\height(\gamma_{[0,\mathbf r_{\delta T}]})\ge T\big)
+\prneinf
\big(\exists n\ge \delta T: 
\height(\gamma_{[0,\mathbf r_n]})=T 
\text{ and
}\width(\gamma_{[0,\mathbf r_n]})\le 2L\big).
\end{align*}
Let $\gamma^{[i]}$ be the $i^{\rm th}$ irreducible bridge of
$\gamma$. Since the $\gamma^{[i]}$s are independent, we obtain 
\begin{multline}
  \tilde B_{T,L}\le \prneinf\big(
\height(\gamma_{[0,\mathbf r_{\delta T}]})\ge T\big)+\prneinf
\big(\forall i\le \delta T, \width(\gamma^{[i]})\le 2L\big)  \\
= \prneinf\big(
\height(\gamma_{[0,\mathbf r_{\delta T}]})\ge T\big)+\prne(\width(\gamma)\le 2L)^{\delta T}  \\
\le \prneinf\big(
\height(\gamma_{[0,\mathbf r_{\delta T}]})\ge T\big)+ \exp \left(-\delta T\,\prne (\width(\gamma)> 2L)\right).\label{ineq}
\end{multline}
Note that 
$$
\height(\gamma_{[0,\mathbf r_{\delta T}]})= \sum_{i=1}^{\delta
  T}\height(\gamma^{[i]}).
$$
Hence the law of large numbers, together with the fact that $\delta
\cdot\mathbb E_{\mathrm{iSAB}}(\height(\gamma))<1$, implies that
$ \prneinf\big(\height(\gamma_{[0,\mathbf r_{\delta T}]})
\ge T\big)$ tends to 0 as $T\rightarrow \infty$. Hence,
if $T\equiv T_k$ and $L\equiv L_k$ are such  that
$T$ and $T\,\prne (\width(\gamma)>2 L)$ both tend to infinity, then
$\tilde B_{T,   L}$ tends to zero. 

\medbreak
We now argue {\it ad absurdum}. Assume that $\mathbb
E_{\mathrm{iSAB}}(\width(\gamma))=\infty$. Then 
$$
\limsup_{L\rightarrow \infty} \frac{\prne(\width(\gamma)>
  2L)}{\frak a_{2L}}=\infty,
$$ 
since $\frak a_L$ is the term of a converging
series (namely, the  \gf\ $A(\x )$ of arches)
and $\prne(\width(\gamma)> L)$ is non-increasing in $L$ and is the term of a
diverging series (indeed, it sums to $\mathbb
E_{\mathrm{iSAB}}(\width(\gamma))=\infty$). 
Let $(L_k)_k$ be a sequence such that 
$$
\lim_{k\rightarrow \infty} \frac{\prne(\width(\gamma)>
  2L_k)}{\frak a_{2L_k}}=\infty,
$$ 
and take 
$$
T_k=\left\lfloor\frac 1 {\sqrt{\frak a_{2L_k}\prne(\width(\gamma)>
  2L_k)}}\right\rfloor.
$$
Then 
$$
T_k\,\prne (\width(\gamma)> 2L_k)\rightarrow \infty\quad\text{ and }
\quad T_k\,\frak a _{2L_k}\rightarrow 0.
$$
It follows {from~\eqref{ineq}}
that $\lim_k\tilde B_{T_k,L_k}=0$. 
But $T_k$ and $L_k$ also satisfy
Lemma~\ref{lem:rectangle_toinfinity_TL}, so we also have $\lim_k\tilde
B_{T_k,L_k}>0$, a contradiction. We 
{thus conclude that }
$\mathbb E_{\mathrm{iSAB}}(\width(\gamma))<\infty$.
\end{proof}

Let $\Omega$ be the set of bi-infinite walks $\gamma: \zs \to \bbH$
such that $\gamma_0=a$. 
Let $(\gamma^{[i]},i\in \mathbb Z)$ be a
bi-infinite sequence of irreducible bridges sampled independently
according to $\prne$. Let $\prneinfbi$ denote the law on $\Omega$
formed by concatenating the bridges $\gamma^{[i]},i\in \mathbb Z\,$ in
such a way that $\gamma^{[1]}$ starts at $a$. Let $\mathcal F$ be the
$\sigma$-algebra generated by events depending on a finite number of
vertices of the walk.  

We extend the {indexing} of renewal points to 
{the bi-infinite walks of $\Omega$}. If $\gamma\in \Omega$ is a
bi-infinite bridge (which is
the case with probability 1 under $\prneinfbi$), we obtain a
bi-infinite sequence $(\mathbf r_n(\gamma))_{n\in \mathbb Z}$ 
such that $r_0(\gamma)=0$. 
Let $\tau:\Omega\rightarrow \Omega$ be the {\em shift} defined by 
$\tau(\gamma)_i=\gamma_{i+\mathbf{r}_1(\gamma)}-\gamma_{\mathbf{r}_1(\gamma)}$
for every $i\in \mathbb Z$.
The shift translates the walk so that $\mathbf{r}_1(\gamma)$
is now at the origin $a$ of the lattice.
Note that $\mathbf{r}_i(\tau(\gamma))=\mathbf{r}_{i+1}(\gamma)
-\mathbf{r}_{1}(\gamma) $.
 Let $\sigma$ denote the reflection in the real axis.

 \begin{Proposition}\label{prop:prop} 
The measure $\prneinfbi$ satisfies the following properties.
\begin{itemize}
\item[{ $(\rm P_1)$}] It is invariant under the shift $\tau$. 
\item[$(\rm P_2)$] The shift $\tau$ is ergodic for $(\Omega,\mathcal
  F,\prneinfbi)$. 
\item[$(\rm P_3)$] Under $\prneinfbi$, the random variables
  $(\sigma\gamma_n)_{n\le 0}$ and $(\gamma_n)_{n\le 0}$ are
  independent and identically distributed. 
\end{itemize} 
\end{Proposition}
 
\begin{proof}
 Property $(\rm P_1)$ is fairly straightforward. Indeed, for every
 $n>0$, the law of 
 $\gamma_{[\mathbf r_{-n}(\gamma),\mathbf r_n(\gamma)]}$ determines,
 in the high-$n$ limit, the law of $\gamma$ (since we work with the
 $\sigma$-algebra $\mathcal F$). Now, the laws of
 $\tau(\gamma_{[\mathbf r_{-n+1}(\gamma),\mathbf r_{n+1}(\gamma)]})$
 and $\gamma_{[\mathbf r_{-n}(\gamma),\mathbf r_n(\gamma)]}$ are the
 same by construction (both are  the law of $2n$
 concatenated independent irreducible  bridges).  Thus $(\rm P_1)$
 follows by letting $n \to \infty$.

Let us turn to $(\rm P_2)$. Consider a shift-invariant event {$\frak A$}. We
want to show that {$\prneinfbi(\frak A)\in\{0,1\}$}. 
Let $\eps>0$. There exists $n>0$ and  an event {$\frak A_n$} depending only on
the vertices $\gamma_{-n},\ldots,\gamma_n$ such that {$\prneinfbi(\frak A_n\,
\Delta \,  \frak A)\le\eps$},
where $\Delta$ denotes the symmetric difference.
In particular, {$|\prneinfbi(\frak A)-\prneinfbi(\frak A_n)|\le \eps$}.
 By extension, {$\frak A_n$} depends only on vertices in
$\gamma_{\mathbf r_{-n}},\ldots,\gamma_{\mathbf r_n}$. Invariance of {$\frak A$}
under $\tau$ implies that 
{$\frak A= \tau^{-2n}(\frak A)$}, so that
\beq\label{A2}
\prneinfbi(\frak A)=
\prneinfbi\left(\frak A\cap \tau^{-2n}(\frak A)\right) \, .
\eeq
Moreover,
\begin{multline*}
 \Big|\prneinfbi\big(\frak A\cap \tau^{-2n}(\frak A)\big) -\prneinfbi \big(\frak A_n\cap \tau^{-2n}(\frak A_n) \big) \Big|\\
  \leq  \prneinfbi\big(\frak A \, \Delta \, \frak A_n \big)  +  \prneinfbi\big( \tau^{-2n}(\frak A) \, \Delta \,  \tau^{-2n}(\frak A_n) \big)  \leq 2 \eps \,.
\end{multline*}
Using~\eqref{A2} and the independence between the walk before and after $\mathbf
r_n$, this reads 
$$
|\prneinfbi(\frak A)-\prneinfbi(\frak A_n)^2|\le 2\eps,
$$
which, combined with {$|\prneinfbi(\frak A)-\prneinfbi(\frak A_n)|\le \eps$},
 implies
$$
|\prneinfbi(\frak A)-\prneinfbi(\frak A)^2|\le 4\eps \, .
$$
By letting $\eps$ tend to $0$, we obtain that
{$\prneinfbi(\frak A)=\prneinfbi(\frak A)^2$} and therefore
{$\prneinfbi(\frak A)\in\{0,1\}$}. Hence $(\rm P_2)$ is proved.

\medbreak
Since the law of irreducible bridges is invariant (up to a
translation) under reflection in  a horizontal line, $(\rm P_3)$ is
straightforward.  
\end{proof}

Renewal points separate a walk into two parts, located below and above
the point. We now introduce a more restrictive notion, illustrated
in Fig.~\ref{fig:diamond} (left).
   A mid-edge $\gamma_k$ of a walk $\gamma$ is said to be a {\em
     diamond point} if  
\begin{itemize}
\item it lies on a vertical edge of the lattice,
\item the walk is contained in the cone 
$$
\textstyle\big((\gamma_k-\frac i2)+\mathbb R_+e^{i\pi/3}+\mathbb R_+
e^{2i\pi/3}\big)\cup\big((\gamma_k+\frac i2)-\mathbb
R_+e^{i\pi/3}-\mathbb R_+ e^{2i\pi/3}\big)
$$
\end{itemize}
(recall that edges have length 1).
The set of diamond points of $\gamma$ is denoted by $\mathbf{D}_\gamma$. 
Of course, it is a subset of $\mathbf{R}_\gamma$. The following
proposition tells us that, under our assumption, a positive fraction
of renewal points {of an infinite bridge}
are diamond points.

\begin{figure}[htb]
\includegraphics[scale=0.4]{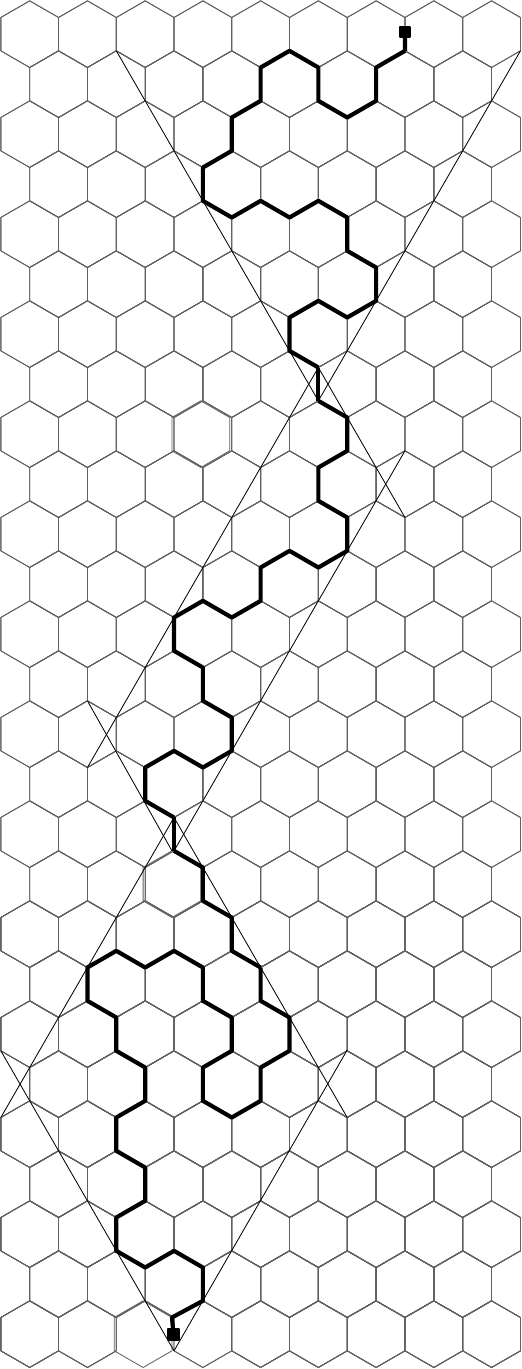}
\includegraphics[scale=0.4]{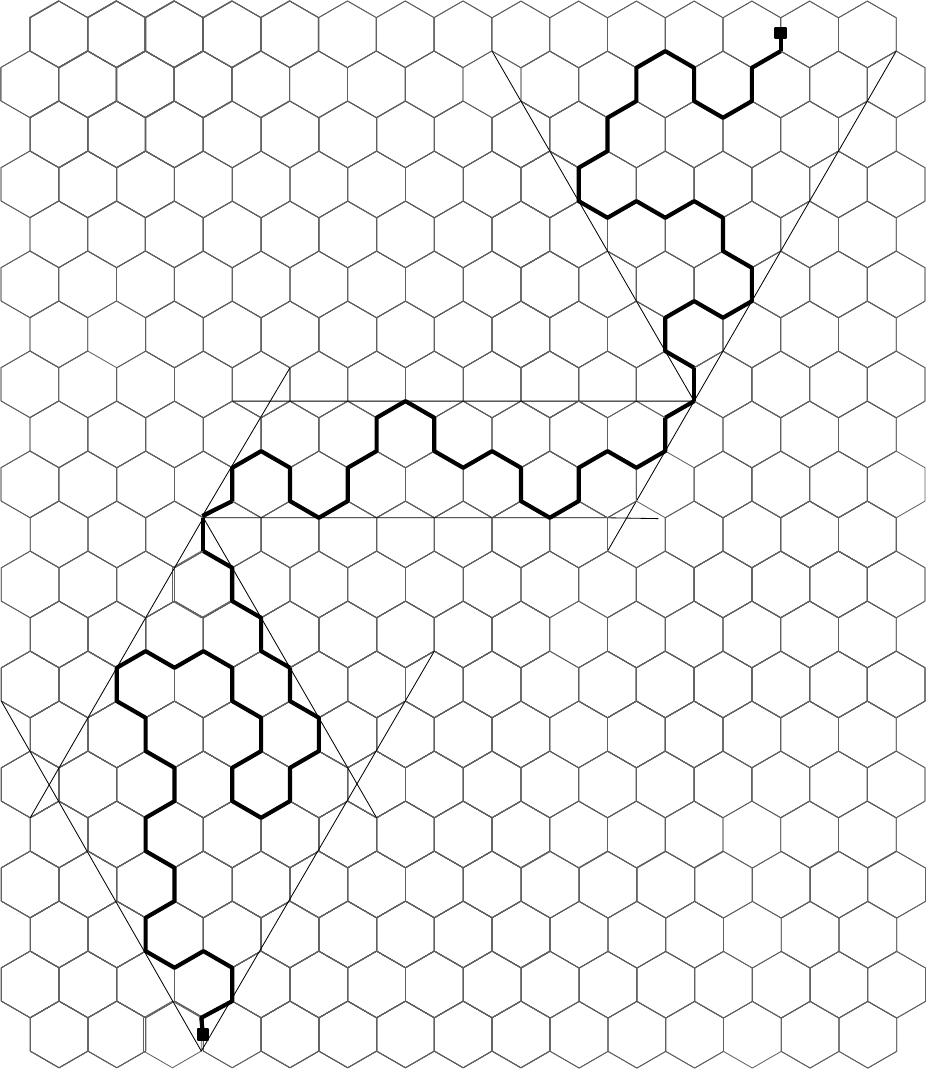}
\caption{{\em Left\/}: A bridge having 
3  diamond points. {\em Right\/}: A stickbreak operation applied to
  this bridge.} 
\label{fig:diamond}\end{figure}

 \begin{Proposition}\label{crucial a}
 If $\mathbb E_{\isab}(\height(\gamma))<\infty$, then
 there exists $\delta>0$ such that
 $$\prneinf\left(\liminf_{n\rightarrow \infty}\frac{|\mathbf{D}_\gamma\cap\{0,\ldots,\mathbf r_n\}|}{n}\ge\delta\right)=1.$$
 \end{Proposition}
 
Let us first provide a heuristic 
argument. Since $\mathbb E_{\isab}(\height(\gamma))$ is finite,
so is $\mathbb
E_{\isab}(\width(\gamma))$ (Proposition~\ref{height width}). 
Then {$\mathbb E_{\isab}(\xco(\gamma_{|\gamma|}))=0$}, and the law of large 
numbers implies that the prefixes of an infinite bridge
are tall and skinny -- that is, height
grows linearly, width grows sub-linearly. So the probability of a 
bridge staying within a cone as thin as one likes is
 positive, and a similar result holds going backwards. 
Thus, diamond points occur with positive density among renewal points.  
 
\begin{proof}
  Let us first prove that $\prneinfbi(\gamma_0\in \mathbf{D}_\gamma)>0$.
Proposition~\ref{height width} shows that $\mathbb
E_{\isab}(\width(\gamma))<\infty$. 
Hence  {$\mathbb E_{\isab}(\xco(\gamma_{|\gamma|}))$} is well-defined, and
is $0$ since
the law of an irreducible bridge is invariant under
 reflection in the imaginary axis.
The law of large numbers thus implies that, $\prneinf$-almost surely, 
{$\xco(\gamma_{\mathbf r_n})/n \longrightarrow 0$}. 
Since the expected width of irreducible bridges is finite, {the
  law of large numbers shows that 
$$ 
\frac{1}{n}\sum_{i=1}^{n}
\width(\gamma_{[\mathbf r_{i-1},\mathbf r_i]}) \longrightarrow
c
\quad \hbox{and} \quad \frac{1}{n+1}\sum_{i=1}^{n+1}
\width(\gamma_{[\mathbf r_{i-1},\mathbf r_i]}) \longrightarrow
c \qquad \text{a.s.,} 
$$ 
where $c:=\mathbb E_{\isab}(\width(\gamma))$ is a positive constant. 
The second identity reads
\[\left(1-\frac{1}{n+1}\right)\frac{1}{n}\sum_{i=1}^{n}
\width(\gamma_{[\mathbf r_{i-1},\mathbf
  r_i]})+\frac{1}{n+1}\width(\gamma_{[\mathbf r_{n},\mathbf r_{n+1}]}) \longrightarrow c<\infty \qquad \text{a.s.,}\]
and 
comparing with the first identity shows that
$\width(\gamma_{[\mathbf r_n ,\mathbf    r_{n+1}
]})/n\longrightarrow0$ almost surely.}
Thus
$$
\frac1n\left(
\frac 1 {\sqrt 3}| \xco(\gamma_{\mathbf r_n
})|+\width(\gamma_{[\mathbf r_n
,\mathbf r_{n+1}
]})\right)\longrightarrow0 \hskip 10mm \hbox{a.s}.
$$
Since
$$
{\width(\gamma_{[0,\mathbf r_n
]})\le2 \max \left\{
\frac 1 {\sqrt 3}
|\xco(\gamma_{\mathbf r_k})|+\width(\gamma_{[\mathbf r_k
,\mathbf r_{k+1}]}),\  0\le k\le n-1\right\},}
$$
we find that, $\prneinf$-almost surely, $\width(\gamma_{[0,\mathbf r_n
]})/n\longrightarrow 0$.

\medskip
Let us now apply the law of large numbers  to 
{$\yco(\gamma_{\mathbf r_n})$}. We obtain  that,   $\prneinf$-almost surely,
{$\yco(\gamma_{\mathbf r_n})/n\longrightarrow  \frac 3 2 \mathbb
  E_{\isab}(\height(\gamma))>0$}.  

Now define
\[I(\gamma):=\inf_{k\ge 0}\left(\yco(\gamma_k)-\sqrt{3}|\xco(\gamma_k)|+1/2\right).\]
Note that for an infinite bridge  $\gamma=(\gamma_0, \gamma_1,
\ldots)$, the origin $\gamma_0$ is a diamond point if and only if $I(\gamma)\ge 0$.

By comparing a general point $\gamma_k$ of $\gamma$ with the last
  renewal point {$\mathbf r_n$}
before $\gamma_k$ and the next one after $\gamma_k$, one finds
\[
I(\gamma) \geq \inf_{n\geq 0}\left(\yco(\gamma_{\mathbf r_n}) -
  \sqrt{3}\left(|\xco(\gamma_{\mathbf r_n})| +\sqrt{3}
    \width(\gamma_{[\mathbf r_n,\mathbf r_{n+1}]})\right) +
  1/2\right).
\]
Then by the arguments presented earlier in this proof, it follows that
{$I(\gamma)>-\infty$} almost surely.

Let $K \in \ns$ be such that 
  $\rho_K:=\prneinf(I(\gamma)\ge -K) > 0$. 
We are going to  show that
\begin{equation}\label{eq:rhorho}
\rho_0  \ge (2\x ^4)^{K} \rho_K >0.
\end{equation}
To prove this, consider an experiment under which the law
$\prneinf$ is constructed by first concatenating $K$ independent samples of
$\prne$ (starting from $a$) and then
an independent sample $\gamma'$ of $\prneinf$. If each of the $K$
samples happens to be a walk of length 4 
 going from $a$ to $a+3i$ and $I(\gamma')\ge -K$, 
then the complete walk $\gamma$ satisfies 
$I(\gamma)\ge 0$. 
The probability that the $i$th sample of $\prne$ is a walk
of length 
4  going from $a$ to $a+3i$ is $2\x ^{4}$. Thus, the
experiment behaves as described with probability $(2\x ^4)^{K}
\rho_K$, and we obtain~\eqref{eq:rhorho}, 
that is, $\prneinf( \gamma_0\in \mathbf{D}_\gamma)>0$.

Using Property $(\rm P_3)$ of Proposition~\ref{prop:prop}, we deduce that
$$
\delta:=\prneinfbi(\gamma_0\in \mathbf{D}_\gamma)
=\left(\prneinf(\gamma_0\in \mathbf{D}_\gamma)\right)^2
>0 .
$$
The shift $\tau$ being ergodic (cf. Property $(\rm P_2)$ of
Proposition~\ref{prop:prop}), 
the ergodic theorem, applied to $\mathbbm 1_{\gamma_0 \in
  \mathbf{D}_\gamma}$, 
gives
$$
\prneinfbi\left(\lim_{n\rightarrow \infty}\frac{|\mathbf{D}_\gamma\cap\{0,\dots,\mathbf r_n(\gamma)\}|}{n}= \delta\right)=1 \, .
$$
Let $\gamma$ be a bi-infinite bridge, and denote $\gamma^+=\gamma_{[0,
  \infty)}$. Then for $n \ge 0$, $\mathbf r_n(\gamma)=\mathbf
r_n(\gamma^+)$, and
$$
\mathbf{D}_\gamma\cap\{0,\dots,\mathbf r_n(\gamma)\}
=\mathbf{D}_\gamma\cap\{0,\dots,\mathbf r_n(\gamma^+)\}
\subset \mathbf{D}_{\gamma^+}\cap\{0,\dots,\mathbf r_n(\gamma^+)\}
$$
since all diamond points of $\gamma$ are diamond points of $\gamma^+$.
This implies that 
\begin{multline*}
\prneinf\left(\liminf_{n\rightarrow \infty}\frac{|\mathbf{D}_{\gamma
}\cap\{0,\dots,\mathbf r_n(\gamma)\}|}{n}
\ge \delta\right)
\\=\prneinfbi\left(\liminf_{n\rightarrow \infty}\frac{|\mathbf{D}_{\gamma^+
}\cap\{0,\dots,\mathbf r_n(\gamma)\}|}{n}
\ge \delta\right)\\
\ge\prneinfbi\left(\liminf_{n\rightarrow \infty}\frac{|\mathbf{D}_{\gamma
}\cap\{0,\dots,\mathbf r_n(\gamma)\}|}{n}\ge \delta\right)=1 .
\end{multline*}
This concludes the proof of the proposition.
\end{proof}

\medbreak
We now introduce some final definitions and a minor lemma before proving the main result.
By Lemma~\ref{ren sar}, we 
want to prove that $\mathbb E_{\textrm{iSAB}}(\height(\gamma))
=\infty$. We {will}
argue {\it ad absurdum}. {Henceforth, assume} $\mathbb E_{\textrm{iSAB}}(\height(\gamma))
<\infty$ and let $\nu>\mathbb
E_{\textrm{iSAB}}(\height(\gamma))$. Also,  let  
$0<\eps<\delta/20$,
where $\delta$ satisfies Proposition~\ref{crucial a}. 

\medbreak
Let $\Omega^+$ denote the set of 
semi-infinite walks in the upper half-plane. That is, $\phi=(\phi_0,
\phi_1, \ldots) \in \Omega^+$ if and only if $\mathbf y(\phi_i) > 0$ 
for $i > 0$. For $\phi \in \Omega^+$ and $\gamma$ a
 finite bridge, we {write} $\gamma\,\triangleleft\, \phi$ if
 $\phi_{[0,|\gamma|]}=\gamma$ and 
 $\phi_{| \gamma |}$ is a renewal
 point of $\phi$. Note that 
\begin{equation}\label{eqtrileft} 
\x ^{|\gamma|}=\prneinf(\phi\in \Omega^+: \gamma\triangleleft \phi) \, .
\end{equation}

Let  $\overline{\sabset}_n$
 denote the set of finite bridges 
$\gamma$ with exactly $n+1$ renewal points (meaning that $\mathbf
r_n(\gamma)=|\gamma|$) such that \begin{itemize}
\item[$({\rm C}_1)$] $\height(\gamma) \le  \nu  n$,
\item[$({\rm C}_2)$] $|\mathbf{D}_\gamma|\ge \delta n/2$.
\end{itemize}
Let us define $\overline{\sabset}_n^+=\{\phi\in \Omega^+
:\exists \gamma\in \overline{\sabset}_n\text{ such that
}\gamma\triangleleft \phi\}$. 
That is, the prefix of $\phi$ consisting of its $n$ first irreducible
bridges satisfies $({\rm C}_1)$ and $({\rm C}_2)$.
It follows from~\eqref{eqtrileft} that
\begin{equation}\label{sarequiv}
\prneinf\big(
\overline{\sabset}_n^+\big)=\sum_{\gamma\in\overline{\sabset}_n}\x
^{|\gamma|} \, . 
\end{equation}
\begin{Lemma}\label{lem:SABplus1}
{Under the above assumptions}, we have,
as $n\to\infty$,
$$ 
\prneinf\big(\overline{\sabset}_n^+\big)\longrightarrow1.
$$ 
\end{Lemma}
\begin{proof}
We consider Conditions $({\rm C}_1)$ and  $({\rm C}_2)$ separately.
Condition $({\rm C}_1)$ for $\gamma\in \overline{\sabset}_n$ 
translates for $\phi\in \overline{\sabset}_n^+$ into
$\height( \phi_{[0, \mathbf r_n]})\le \nu n$.
Since $\mathbb E_{\textrm{iSAB}}(\height(\gamma)) <\nu$, the law of
large numbers gives
$$
\prneinf\Big(\phi\in\Omega^+:
\height( \phi_{[0, \mathbf r_n]})\le \nu n\Big)\longrightarrow 1.
$$
Let us now consider Condition $({\rm C}_2)$, which translates into
$|\mathbf D_{\phi_{[0,\mathbf r_n]}}|\ge \delta n/2$.
But 
$$
\mathbf D_{\phi_{[0,\mathbf r_n]}}\supset \mathbf D_\phi
\cap\{0, \ldots, ,\mathbf r_n\},
$$
since the truncation
operation $\phi \to \phi_{[0,\mathbf r_n]}$ can only
create (and not annihilate) diamond points.
Thus Proposition~\ref{crucial a} yields 
$$
\prneinf\big(|\mathbf D_{\phi_{[0,\mathbf r_n]}}| \ge \textstyle\frac\delta2 n\big) \longrightarrow 1 ,
$$
and we have proved the lemma.
\end{proof}

\medbreak

\begin{proof}[Proof of Theorem~{\rm\ref{thm:BT0}}.]
We are going to prove that
\begin{equation}\label{eq:littlecontra}
\prneinf\Big(\width(\phi_{[0,\mathbf r_{\nu n+1}]})> \eps n \Big)
\ge \left(\frac{\delta n \x }{
10(\nu n+2)
}\right)^2\, \prneinf\big(\overline{\sabset}_n^+\big) \, .
\end{equation}
We proved at the beginning  of the proof of Proposition~\ref{crucial
  a} that $\width(\phi_{[0,\mathbf r_{\nu n +1}(\phi)]})/n$ tends to
zero almost surely under the $\prneinf$-distribution. Thus the
left-hand side of the above inequality tends to $0$ as $n \rightarrow
\infty$, and so does the right-hand side.
This contradicts {Lemma~\ref{lem:SABplus1}}
and proves that our assumption
$\mathbb E_{\textrm{iSAB}}(\height(\gamma))<\infty$ cannot hold.

\medbreak
Consider $\gamma\in \overline{\sabset}_n$. Let $\mathbf{d}_i$
be the index of the $i$th diamond point of $\gamma$. For integers $i\in
\big[\frac\delta {10} n,\frac {2\delta} {10} n\big]$ and
$j\in\big[\frac {3\delta}{10}n,\frac{4\delta}{10}n\big]$,
let $\stickbreak_{i,j}(\gamma)$ be the following walk (see
Fig.~\ref{fig:diamond}, right):
\beq\label{def-sb}
\stickbreak_{i,j}(\gamma) = \gamma_{[0,\mathbf{d}_i
]} \circ s  \circ \rho
\big(\gamma_{[\mathbf{d}_i
,\mathbf{d}_{j}
]}\big) \circ \tilde s  \circ 
\gamma_{[\mathbf{d}_{j}
,\mathbf{r}_n ]} \, ,
\eeq
where $\circ$ stands for the concatenation of walks, $\rho$ is the clockwise
rotation of angle $\pi/3$, 
$s $ is a single right turn,
and $\tilde s $ is a single  left turn.
The definition of diamond points implies that $\stickbreak_{i,j}(\gamma)$
is not only self-avoiding, but also a bridge. Also, note that we used
$({\rm C}_2)$ in order to define $\stickbreak(\gamma)$ for all these
values of $i$ and $j$. 

Let 
$$
\Phi~=~\big[\textstyle\frac{\delta}{10}
n,\frac{2\delta}{10}n\big]~\times~
\big[\frac{3\delta}{10}n,\frac{4\delta}{10}n\big]~\times~
\overline{\sabset}_n \, , 
$$ 
and denote 
$$
S:=\sum_{(i,j,\gamma)\in \Phi}\x ^{|\stickbreak_{i,j}(\gamma)|} \, .
$$
One can express $S$ in terms of $\prneinf\big(
\overline{\sabset}_n^+\big)$. Indeed,
$\vert \stickbreak_{i,j}(\gamma) \vert = \vert \gamma \vert + 2$, and therefore
\begin{equation}\label{eqsinferone}
S = \sum_{(i,j,\gamma)\in \Phi}\,\x ^{|\gamma|+2} = \left(\frac{\delta
    n \x 
    }{10}\right)^2\sum_{\gamma\in \overline{\sabset}_n}\x ^{|\gamma|}
= \left(\frac{\delta n \x }{10}\right)^2\,\prneinf\big( \overline{\sabset}_n^+\big) \, .
\end{equation}
We used \eqref{sarequiv} for the last equality. 
We are now going to give an upper bound on $S$, which will
imply~\eqref{eq:littlecontra}.
\medbreak

Note that the walk $\gamma_{[\mathbf{d}_i,\mathbf{d}_{j}]}$ contains at
  least $\delta n/10$ diamond points, and thus has height $h:=
  \height(\gamma_{[\mathbf{d}_i,\mathbf{d}_{j}]}) \ge \delta
  n/10$. Rotating this walk {clockwise} by $\pi/3$ results in a
  walk of height at most $h$ and width at least {$h/2$ }
(due to the honeycomb geometry and the way we define height and width just
    above Lemma~\ref{ren sar};
{the extreme case is reached for the bridge obtained by concatenating
$h$ copies of the bridge formed by a left turn followed by a right turn). 
}%
Since the width of $\stickbreak_{i,j}(\gamma)$ must be greater
  than the width of any of its subwalks, in particular greater than
  the width of $\rho(\gamma_{[\mathbf{d}_i,\mathbf{d}_{j}]} )$, we
  must have 
\[
\width(\stickbreak_{i,j}(\gamma)) \geq \frac{h}{2} \geq \frac{\delta
  n}{20} > \eps n,
\]
{since we have assumed $\eps <\delta/20$.}
{Since $\height(\rho(\gamma_{[\mathbf{d}_i,\mathbf{d}_{j}]})) \leq
  \height(\gamma_{[\mathbf{d}_i,\mathbf{d}_{j}]})$, the $\stickbreak$
  operation increases the height of $\gamma$ by at most~1}
{(due to the attachment of the steps $s$ and $\tilde s$).}
By $({\rm C}_1)$, we {then} have $\height(\stickbreak_{i,j}(\gamma))\le
\nu n
+1$ and therefore $\stickbreak_{i,j}(\gamma)$ has at most
 $\nu n +2$ renewal points {(as any subwalk between two renewal points must have height at least 1)}. Hence, for any $\phi\in \Omega^+$ such that
 $\stickbreak_{i,j}(\gamma)\triangleleft \phi$, we 
have $\mathbf r_{\nu   n +1} \geq | \stickbreak_{i,j}(\gamma) |$ 
and therefore
 $$
\width(\phi_{[0,\mathbf{r}_{\nu n+1}]})\ge
{\width(\stickbreak_{i,j}(\gamma))} >
\eps n.
$$ 
Thus, for any $(i,j,\gamma)\in \Phi$, 
\begin{align*}
\x^{|\stickbreak_{i,j}(\gamma)|} & = \prneinf\big( \phi\in \Omega^+:\stickbreak_{i,j}(\gamma)\triangleleft \phi \big) \\
& = \prneinf\big( \phi\in\Omega^+:\stickbreak_{i,j}(\gamma)\triangleleft \phi\text{ and }\width(\phi_{[0,\mathbf{r}_{\nu n+1}]})> \eps n \big).
\end{align*}
Therefore,
\begin{eqnarray} 
S &= & \sum_{(i,j,\gamma)\in \Phi}\prneinf\big(\phi\in
\Omega^+:\stickbreak_{i,j}(\gamma)\triangleleft \phi\ \text{and}\
\width(\phi_{[0,\mathbf{r}_{\nu n+1}]})> \eps n\big)
\nonumber \\ 
&=&\mathbb E_{\rm iSAB}^{\otimes\mathbb N}\Big(\big| \big\{
(i,j,\gamma)\in \Phi:\stickbreak_{i,j}(\gamma)\triangleleft\phi \big\}
\big| \cdot \mathbbm{1}_{\{\width(\phi_{[0,\mathbf{r}_{\nu
      n+1}]})>\eps n\}}\Big) \nonumber \\ 
& \le & 
(\nu n +2)^2\,\prneinf\big(
\width(\phi_{[0,\mathbf{r}_{\nu n+1}]})> \eps n
\big).\label{eqsinfertwo}
\end{eqnarray}
 The last inequality follows from the fact 
 that, for any given  $\phi \in \Omega^+$, the number of elements $(i,j,\gamma)$ of $\Phi$
such that $\stickbreak_{i,j}(\gamma)\triangleleft\phi$ is at most $(
\nu n+2)^2$. Indeed, 
  the triple $(i,j,\gamma)$ is completely determined if we
 specify in $\phi$ the renewal point that precedes the step denoted
 $s $ in~\eqref{def-sb} and
 the one that follows the step $\tilde s $. As both points occur before
 $\mathbf{r}_{\nu n+1}$, as explained above, the bound~\eqref{eqsinfertwo} follows.

By combining  \eqref{eqsinferone} and  \eqref{eqsinfertwo}
we obtain~\eqref{eq:littlecontra}, which  concludes the proof.
\end{proof}

\bibliographystyle{plain}
\bibliography{hexa}

\end{document}